\RequirePackage{xr-hyper}
\documentclass[sigplan,10pt]{acmart}
\def\paperwrapper{1}
\def\fullpaper{1}
\ifdefined\paperwrapper\else
\RequirePackage{xr-hyper}
\documentclass[sigplan,10pt]{acmart}
\fi
\renewcommand\footnotetextcopyrightpermission[1]{}
\usepackage{amsmath,amssymb,mathtools}
\usepackage{booktabs}
\usepackage{multirow}
\usepackage{siunitx}
\usepackage{algorithm}
\usepackage{algpseudocode}
\usepackage{cryptocode}
\usepackage{pgfplots}
\usepackage{xspace}
\usepackage[leftcaption]{sidecap}
\pgfplotsset{compat=1.18}
\usepackage{enumitem}
\usepackage{pifont} 

\usepackage{tikz}
\usetikzlibrary{positioning, fit, backgrounds, arrows.meta, calc}

\usepackage[most]{tcolorbox}
\definecolor{tkshade}{rgb}{0.995,0.975,0.90}
\definecolor{tkgold}{rgb}{0.80,0.60,0.10}
\definecolor{tkbrown}{rgb}{0.42,0.30,0.04}
\newtcolorbox{takeaway}[1][Takeaway]{%
  enhanced, breakable=false,
  colback=tkshade, colframe=tkgold!80!tkbrown,
  coltitle=tkbrown, colbacktitle=tkgold!20!white,
  fonttitle=\bfseries\small, fontupper=\small\itshape,
  sharp corners=west, rounded corners=east,
  boxrule=0.8pt, borderline west={2pt}{0pt}{tkbrown},
  left=2mm, right=2mm, top=0.6mm, bottom=0.6mm, boxsep=1mm,
  before skip=4pt, after skip=4pt,
  title={#1},
}

\newif\ifcomment
\commenttrue

\ifcomment
    \newcounter{YXNumberOfComments}
    \stepcounter{YXNumberOfComments}
    \newcommand{\xym}[1]{\textcolor{orange}{\small \bf [XYM\#\arabic{YXNumberOfComments}\stepcounter{YXNumberOfComments}: #1]}}

    \newcounter{PHNumberOfComments}
    \stepcounter{PHNumberOfComments}
    \newcommand{\hpc}[1]{\textcolor{blue}{\small \bf [HPC\#\arabic{PHNumberOfComments}\stepcounter{PHNumberOfComments}: #1]}}  
    
    \newcommand{\del}[1]{{{\color{red}\st{#1}}}}
\else
    \newcommand\xym[1]{}    
    \newcommand\hpc[1]{} 
    \newcommand{\del}[1]{}
\fi

\newcommand{\share}[1]{\langle #1\rangle}
\newcommand{\bits}{\{0,1\}}
\newcommand{\Z}{\mathbb{Z}}
\newcommand{\sgn}{\operatorname{sign}}
\newcommand{\popcount}{\mathsf{WallacePopcount}}
\newcommand{\lepub}{\mathsf{LE}_{\mathrm{pub}}}
\newcommand{\sys}{\textsc{Spruce}\xspace}

\newif\ifincludeappendix
\ifdefined\appendixonly\includeappendixtrue\fi
\ifdefined\fullpaper\includeappendixtrue\fi

\ifdefined\fullpaper
\newcommand{\suppref}[1]{\ref{#1}}
\newcommand{\paperref}[1]{\ref{#1}}
\else
\newcommand{\suppref}[1]{\ref{supp-#1}}
\newcommand{\paperref}[1]{\ref{paper-#1}}
\fi

\begin{document}

\ifdefined\appendixonly
\externaldocument[paper-][nocite]{full}[main.pdf]
\title[Supplementary Appendix: \sys]{Supplementary Appendix for\\ \sys: Representation-Driven Scaling for Private Outsourced Retrieval}
\author{Submission \#142}
\maketitle
\thispagestyle{plain}
\pagestyle{plain}
\else
\ifdefined\fullpaper\else
\externaldocument[supp-][nocite]{full}[appendix.pdf]
\fi
\title{\sys: Scalable Private Outsourced Retrieval Using Compact Embeddings}

\author{Peichun Hua}
\affiliation{%
  \institution{The Chinese University of Hong Kong, Shenzhen 518172} \institution{State Key Laboratory of Internet Architecture, Tsinghua University, Beijing 100084}
  \country{China}
}
\email{peichunhua@link.cuhk.edu.cn}

\author{Yunming Xiao}
\authornote{Corresponding author.}
\affiliation{%
  \institution{The Chinese University of Hong Kong, Shenzhen 518172}
  \institution{State Key Laboratory of Internet Architecture, Tsinghua University, Beijing 100084}
  \country{China}
}
\email{yunmingxiao@cuhk.edu.cn}

\settopmatter{authorsperrow=2}
\begin{abstract}
Retrieval-Augmented Generation (RAG) has made dense retrieval over large document collections a default building block, and organizations increasingly outsource the vector index to untrusted clouds—exposing both proprietary corpora and user queries. Protecting this retrieval cryptographically is difficult because every query searches corpus-scale state: computation, correlated randomness, and communication grow with every indexed document. At million-document scale, a naive secure implementation takes minutes and $\sim$90\,GB of communication per query. Recent optimized systems take 10--22 seconds per query, but this is still comparable to downstream LLM generation and therefore a first-order contributor to end-to-end latency.

We propose \sys 
(\underline{S}calable \underline{P}rivate Outsourced \underline{R}etrieval \underline{U}sing \underline{C}ompact \underline{E}mbeddings)
, which co-designs the representation with the cryptographic protocol.
\sys learns compact binary hash codes that preserve the candidates needed for subsequent full-precision reranking, replacing corpus-wide full-embedding scoring with efficient Hamming-distance computation under two-server multi-party computation (MPC). 
A corpus-calibrated fixed-radius protocol avoids multi-round candidate selection while preserving the final retrieval quality.
\sys further supports two optimizations that target deployment bottlenecks: private cluster pruning trades minor quality loss for drastically less computation, while a one-core owner-operated dealer removes cloud OT as a preprocessing bottleneck.
Across four corpora spanning 383K–5.42M documents, \sys retains the original search quality with median candidate sets of merely 382–1,952. At 10\,Gbps inter-server bandwidth, its full scan takes 0.21–2.97 seconds, $4.8$–$6.7\times$ faster than the closest measured prior work, while private pruning takes 0.06–1.09 seconds, $13.1$–$22.9\times$ faster, and retains $93.9$–$97.3\%$ of full-float NDCG. On the largest corpus, pruning and the dealer jointly raise sustained throughput by $31.5\times$ at 1\,Gbps per-link bandwidth.
\end{abstract}



\maketitle
\thispagestyle{plain}
\pagestyle{plain}

\section{Introduction}

Retrieval-Augmented Generation (RAG) has turned dense information retrieval over large corpora into core infrastructure for knowledge-grounded language applications~\cite{lewis2020rag,gao2023retrieval,fan2024survey}. A RAG system answers a query by first \emph{retrieving}: it embeds the query with a neural encoder, scores that embedding against a precomputed index of document embeddings, and returns the top-$k$ documents, which are concatenated into the prompt of a generative model~\cite{lewis2020rag,gao2023retrieval}. As the corpus grows, the cost of retrieval grows with it, motivating more organizations to increasingly outsource retrieval to managed cloud services~\cite{chang2025remoterag,ming2026p2rag,hua2026semantics}. 

This is convenient, but from a confidentiality standpoint, it is alarming: the cloud now holds both a proprietary corpus—often the asset that the organization wants to protect the most—and a stream of user queries that reveal private intent \cite{huang2023privacy,zeng2024good}. Prior embedding inversion attacks and attribute inference attacks have shown that dense vectors are not opaque identifiers but are invertible semantic footprints \cite{morris2023text,li2023geia,song2020information,chen2025algen}, making both document indexes and query representations privacy-sensitive. 
At the system level, privacy protection in RAG spans two distinct components: neural model computation and outsourced retrieval state. Existing private-inference systems primarily address the former by securing transformer forward passes~\cite{hao2022iron,pang2024bolt,lu2023bumblebee,xu2025blb}, whose cost is governed by model size rather than corpus size. 

This paper targets the latter. 
We assume that the client runs the query encoder locally, while the cloud stores the document embeddings and contents 
and performs retrieval for multiple users in an organization. At the scale of $10^7$ documents, this state can occupy tens of gigabytes or more,  and searching it can require substantially more computation than encoding a query.

The challenge is to search this large outsourced state without exposing either the corpus or the encoded query, while keeping the cost practical at corpus scale. Since downstream LLM generation in RAG already takes seconds,
a practical private service should keep retrieval within this seconds-scale latency budget while 
sustaining multi-user throughput. 

Among available approaches, secret-sharing-based multi-party computation (MPC) is particularly attractive because it offers more efficient similarity computation than homomorphic encryption (HE) while jointly protecting the corpus and query; however, a full scan still incurs corpus-linear computation, correlated randomness, and communication (\S\ref{sec:mpc-bg}). 
Table~\ref{tab:motivation} shows the challenge: even secure int8 cosine over $N$ embeddings of dimension $D{=}768$ requires $N\!\cdot\!D$ secure multiplications, taking $\sim\!9.4$ minutes and $\sim\!88$\,GB per query at one million documents.

Recent privacy-preserving RAG systems approach this cost at different points. $p^2$RAG~\cite{ming2026p2rag} is closest to our setting: like us, it secret-shares the outsourced embeddings across two non-colluding servers and optimizes candidate selection through interactive bisection; however, it still securely scores every full embedding.
Other systems assume that the retrieval provider owns or is trusted with the corpus and therefore focus on query privacy: RemoteRAG leverages differential privacy (DP) to perturb the query before PHE scoring, PANTHER combines clustering with PIR and MPC, and Pisces applies SimHash and BM25 filters before secure scoring~\cite{chang2025remoterag,li2025panther,liang2026pisces}.  Despite these optimizations, our evaluation shows that their encrypted reranking, PIR state, or candidate pools remain costly at a million-document scale.

\begin{table}[t]
\centering
\caption{Per-query cost of the \emph{search} step of one private query ($D{=}768$): the int8 cosine baseline (DFP, direct full-precision: the uncompressed $D$-dim embedding scored in int8) vs.\ our learned $L{=}128$ hash, across three corpus sizes.}
\label{tab:motivation}
\small
\setlength{\tabcolsep}{4pt}
\resizebox{\columnwidth}{!}{
\begin{tabular}{l l S[table-format=5.3] S[table-format=5.2] S[table-format=6.1]}
\toprule
Representation & per-query op &
\multicolumn{1}{c}{$N{=}382K$} &
\multicolumn{1}{c}{$N{=}10^6$} &
\multicolumn{1}{c}{$N{=}10^7$} \\
\midrule
int8 cosine (DFP) & mults ($\times10^8$)        & 2.9   & 7.7    & 76.8 \\
\quad online comm (MB)   &                       & 33611 & 87863  & 878628 \\
\quad online latency (s) &                       & 215.0 & 562.1  & 5621 \\
\addlinespace
hash $L{=}128$ (ours) & ANDs ($\times10^8$)      & 0.5   & 1.3    & 12.8 \\
\quad online comm (MB)   &                       & 32.9  & 86.0   & 860 \\
\quad online latency (s) &                       & 0.183 & 0.482  & 4.85 \\
\bottomrule
\end{tabular}%
}
\end{table}


\paragraph{Our approach: \sys.} 
Our key insight is to co-design the retrieval representation with the secure protocol, so that corpus-wide secure computation operates only on compact learned binary codes, while full-precision scoring is restricted to a small candidate set. 
\sys therefore separates private retrieval into a corpus-wide secure filter and a candidate-only exact stage. At setup, the corpus is encoded once into compact codes and full-precision embeddings, then split across two non-colluding servers. At query time, the servers reveal a coarse candidate set from the compact codes; the trusted client reranks the corresponding embeddings and fetches the final documents obliviously.

\paragraph{Coarse filtering on binary representations.} The \textsf{Filter} scores learned $L$-bit codes by Hamming distance—the number of differing bits, computed by a bitwise XOR followed by a popcount~\cite{wang2014hashing,luo2023survey}. This maps efficiently to two-server MPC: XOR-shared bits are combined locally without interaction, so only the popcount and one distance-to-radius comparison consume interactive AND gates. The filter thus replaces the $N\!\cdot\!D$ secure multiplications of full-embedding cosine with roughly $L\!\cdot\!N$ Boolean ANDs (\S\ref{sec:protocol}), making $L$ a direct gate budget. Our lightweight deep-hashing recipe (\S\ref{sec:why-binary},~\S\ref{sec:hash-arch}) learns compact codes that beat the $768$-bit sign-of-float baseline~\cite{reimers2019sentence} at a fraction of its gate count.

\paragraph{A protocol with calibrated radius to tame interaction.} The \textsf{Filter} must turn shared Hamming distances into a candidate set. Exact top-$K$ selection would need an oblivious sort over the shares~\cite{hamada2012sorting,bogdanov2014sorting}, which is expensive. A radius threshold is much cheaper, but choosing the radius online from the encrypted distance distribution with binary search costs $\log_2(L{+}1)$ data-dependent count reveals and additionally leaks a sketch of the corpus distance CDF. \sys moves this choice offline: repeated stratified calibration finds the smallest radius that reaches a target final NDCG after float reranking, and the resulting corpus-level radius is supplied up front, collapsing online candidate selection to a \emph{single} comparison and a \emph{single} reveal (\S\ref{sec:protocol}) while directly optimizing the retrieval metric the application consumes.

\paragraph{Additional deployment optimizations.} To further facilitate practical deployments, \sys additionally supports two optimizations. First, users can privately retrieve a fixed set of padded Hamming clusters before MPC and trade a slight quality loss for significantly lower work. Second, organizations that can run one small trusted server can seed triple generation and eliminate a critical communication bottleneck among the two servers. The two optimizations compose directly, and can be disabled independently (\S\ref{sec:opts}).

\paragraph{Oblivious fetch and a bounded leakage profile.} Once the calibrated \textsf{Filter} reveals a small but coarse candidate set, the client reconstructs only the candidate embeddings, reranks them against its full-precision query embedding, and keeps the top-$k$. Document content is stored as client-encrypted ciphertext replicated to both servers, and the client retrieves its chosen blobs by two-server PIR~\cite{chor1998private}, so the servers never learn which documents are returned (\S\ref{sec:fetch}). We give an adaptive multi-query simulation proof for explicit setup and online leakage functions and quantify the full-scan access pattern with a ciphertext-only co-occurrence estimator (\S\ref{sec:leakage}). The experiment recovers a low but measurable unlabeled neighborhood signal.

\paragraph{Results.} Across four BEIR corpora spanning 383K--5.42M documents, \sys retains $95.2$--$97.8\%$ of full-corpus float NDCG with median candidate set size of 382–1,952. At 10\,Gbps inter-server bandwidth, its full scan runs $4.8$--$6.7\times$ faster than the closest measured prior path~\cite{ming2026p2rag,chang2025remoterag,liang2026pisces}; private pruning raises this advantage to $13.1$--$22.9\times$ while retaining $93.9$--$97.3\%$ of full-float NDCG. On the largest corpus, pruning and the dealer jointly raise sustained throughput by $31.5\times$ when inter-party links are capped at 1\,Gbps.

\section{Background and Threat Model}
\label{sec:background}

\subsection{Dense Retrieval}
A dual-encoder retriever maps a query and each document into a shared $D$-dimensional space and ranks them by inner product or cosine similarity~\cite{karpukhin2020dpr,wang2022e5,chen2024bge}. Modern retrievers fine-tune a pretrained transformer with a contrastive objective and hard-negative mining, and generalize zero-shot across domains; this is why a single encoder serves heterogeneous corpora~\cite{reimers2019sentence,thakur2021beir}. At serving time, the document embeddings form a static index of size $N{\times}D$; a query is encoded once and scored against the entire index, and the top-$k$ results are returned and, in RAG, concatenated into the generator's prompt~\cite{lewis2020rag,wang2024biorag,ram2023context}. Exact scan is $O(ND)$ per query; production systems, instead, build an approximate-nearest-neighbor index such as HNSW, IVF, and Product Quantization~\cite{malkov2018efficient,jegou2011ivfpq,douze2025faiss} to sublinearize the search, at the cost of higher storage overhead or loss of search quality. 

\paragraph{Retrieval metrics.} Recall@$k$ measures the fraction of a query's gold-relevant documents returned in the top $k$, while NDCG@$k$ weights graded relevance by rank and normalizes against the ideal ranking. Because our filter emits a candidate set rather than a final ranking, we report candidate recall $R_{\mathrm{float}@10}$: the fraction of the full-corpus float top 10 present anywhere in that set. $R_{\mathrm{float}@10}$ measures fidelity to the float retriever, whereas final NDCG@10 measures the quality of the float-reranked output against gold relevance labels.

\paragraph{Challenges in private search} Unfortunately, none of the index types above survive transplantation into a cryptographic backend unchanged. For example, the data-dependent traversal in HNSW is exactly the access pattern that an oblivious protocol must hide~\cite{zhu2025compass,cui2026mess}. When the index is outsourced, the cloud sees the entire $N{\times}D$ matrix of sensitive embeddings and every query embedding, which inversion attacks may turn back into text~\cite{morris2023text,li2023geia}.
Both the cost of this scan and the exposure of the index hinge on the representation that the cloud computes. One representation, long used for plaintext efficiency, is the short binary \emph{hash code}, where similarity becomes a Hamming distance, which consists of only a bitwise XOR and a popcount~\cite{do2016bdnn,wang2014hashing,luo2023survey}.
Two of its properties turn out to matter once the computation moves under cryptography: the Hamming distance is among the cheapest operations to evaluate securely (\S\ref{sec:mpc-bg}), and the code length $L$ is a free knob, decoupled from the encoder dimension $D$, so the per-document work can be set independently of the model. \S\ref{sec:hash-arch} details how such codes are learned.

\subsection{Cryptographic Preliminaries}
\label{sec:mpc-bg}
Our backend composes two standard cryptographic building blocks—two-server secure computation and PIR—both in the same non-colluding, semi-honest model.

\paragraph{Two-server secure computation.} We work in the standard two-server (2PC) setting: two servers $P_A$ and $P_B$ execute the protocol honestly but are \emph{curious}—each may inspect its own transcript to infer what it can—and do not collude~\cite{goldreich1987gmw}. Privacy comes from \emph{secret sharing}: every sensitive value is split into two shares, one held by each server, that reconstruct the plaintext while either share \emph{alone} is uniformly random and reveals nothing. A bit $x\in\bits$ is \emph{XOR-shared} as $x=\share{x}_A\oplus\share{x}_B$; to produce the sharing, one samples a uniform bit $r$ and sends $\share{x}_A{=}r$ to $P_A$ and $\share{x}_B{=}r\oplus x$ to $P_B$, so neither share on its own says anything about $x$. 

The cost of a computation is then set by which gates it uses. \emph{Linear} gates (including XOR and NOT) are \emph{free}: to XOR two shared bits, each server XORs its own pair of shares locally, and the local results already recombine to the correct answer, with no communication and no setup. The \emph{nonlinear} gate—AND—cannot be evaluated from local shares and is the expensive primitive. It is computed with a \emph{Beaver triple}~\cite{beaver1991circuit}: a random pre-shared triple $(\share{a},\share{b},\share{c})$ satisfying $c{=}a\wedge b$, generated in an input-independent offline phase and then consumed online to reduce the AND to a few local XORs plus a single round of masked communication. Each AND thus spends one triple, and ANDs at the same circuit depth batch into one communication round. The clouds generate triples before queries through oblivious transfer (OT), where a receiver obtains one of two sender messages without exposing its choice or the other message. Silent OT derives large batches from base OTs with mostly local expansion~\cite{boyle2019silentot,libOTe}; our implementation converts two random OTs into each Boolean triple. Buffered generation is absent from single-query latency, but its supply rate bounds sustained throughput. \S\ref{sec:opts} reduces triple demand or replaces its source with a seeded pseudorandom-generator (PRG) dealer.

The same free-linear / charged-nonlinear split holds in the \emph{arithmetic} domain over $\Z_{2^{32}}$ used for inner products: values are additively shared, additions are local, and each multiplication consumes one arithmetic triple. But the two retrieval metrics load these primitives very differently. A cosine similarity is a sum of products, so its bulk work is the $N\!\cdot\!D$ per-dimension \emph{multiplications}, every one a charged triple; a Hamming distance is an XOR followed by a popcount, so its bulk work—per-bit \emph{XOR} of the query against every hash code—is \emph{free}, with only the popcount spending ANDs.

\begin{figure*}[thbp]
\centering
\includegraphics[width=0.9\textwidth]{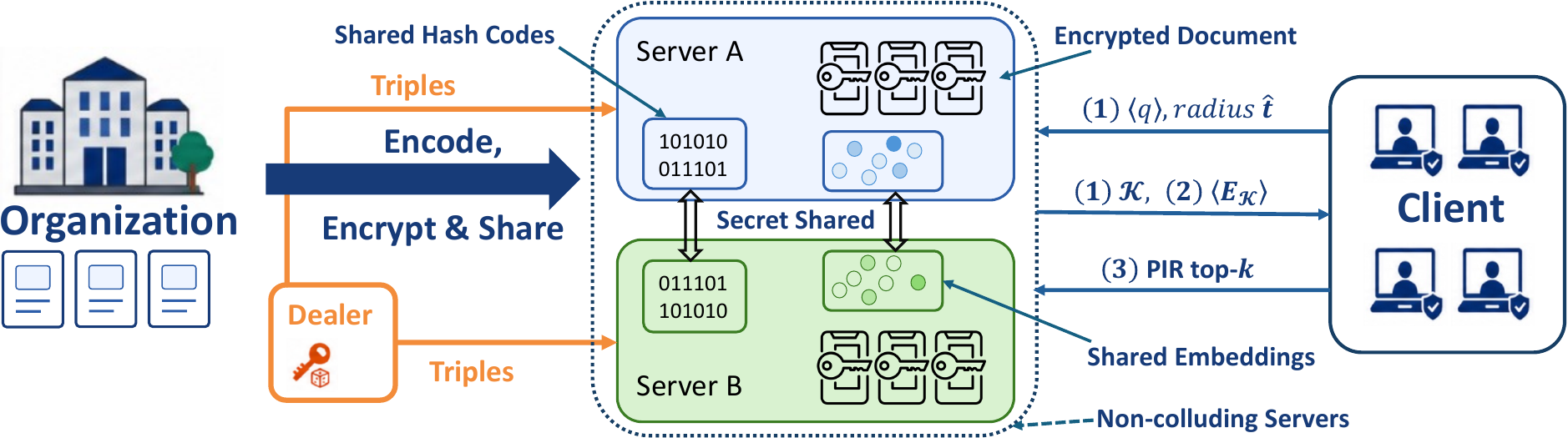}
\caption{\sys overview. During setup, the owner shares codes and embeddings, replicates encrypted content, and optionally supplies triples through a dealer. Online arrows denote (1) Boolean-MPC \textsf{Filter} input/output, (2) candidate-embedding shares for client-side \textsf{Rerank}, and (3) PIR \textsf{Fetch}; private pruning optionally restricts (1) to padded buckets.}
\Description{Architecture diagram with a trusted organization on the left, two non-colluding cloud servers in the center, and a trusted client on the right. The owner supplies the encrypted and shared index, an optional dealer supplies triples, and the query follows Filter, Rerank, and Fetch stages.}
\label{fig:overview}
\end{figure*}

\paragraph{Private information retrieval.} Private information retrieval (PIR)~\cite{chor1998private} lets a client fetch record $i$ from an $n$-record database without revealing $i$~\cite{ostrovsky2007survey,ulukus2022private}. We use the classical information-theoretic two-server construction~\cite{chor1998private}: two non-colluding servers each hold an identical replica of the $n$-record database, and the client splits a one-hot selector for $i$ into two XOR shares. Each server returns the XOR of the records selected by its share; either the client XORs the responses to recover record $i$, or the responses remain XOR shares for a later Boolean MPC. Each server sees a uniformly random selector and learns nothing about $i$. Private pruning uses the shared-output path to load buckets without revealing their IDs to either cloud (\S\ref{sec:private-prune}); final content fetch reconstructs the selected ciphertext at the client (\S\ref{sec:fetch}).

\subsection{Threat Model}
\label{sec:threat}
\sys involves three roles. A \emph{data owner} holds the corpus and, in a one-time offline setup, encodes and indexes it before handing the index to the servers. Two \emph{servers} $P_A$ and $P_B$ jointly store the outsourced index and answer queries. A \emph{client} issues queries and is trusted; it holds the plaintext query and receives the final retrieved documents. 
The two servers are \emph{semi-honest and non-colluding}: each follows the protocol but may inspect its own view to infer what it can, and the two do not share their views. This is the standard model for two-server PIR and a realistic one for commercial two-cloud deployments, where the providers are competitors with no incentive to collude.

The security goal is to keep both the corpus and the user's queries confidential from the servers. An optional trusted owner-operated dealer may supply query-independent correlated randomness (\S\ref{sec:opts}) \textit{without} storing any corpus or index and is not required for security or correctness.

We specify the leakage that \sys does permit: the servers learn a coarse access pattern over the corpus, which we capture with simulation-based leakage functions and evaluate through relational access pattern inference in \S\ref{sec:leakage}. We exclude malicious (actively deviating) or colluding servers, and we do not rely on additional hardware trust such as a server-side TEE.


\section{System Design}
\label{sec:design}

\subsection{System Overview}
\label{sec:overview}

Figure~\ref{fig:overview} shows a high-level demonstration of our architecture. \sys consists of a one-time offline setup run by the data owner, and an online query path with three stages: \textsf{Filter}, \textsf{Rerank}, and \textsf{Fetch}.

\paragraph{Offline setup.} The data owner encodes the corpus once, producing for each document an $L$-bit hash code, a full-precision embedding. The codes $H\in\bits^{N\times L}$ and the embeddings are XOR secret-shared across the two servers, so neither server alone holds a single code bit or embedding byte; the content is encrypted under a client-held key, and the \emph{identical} ciphertext is replicated to both servers. The servers thus store a fully hidden index together with a pair of identical encrypted ciphertext databases, after which the data owner may go offline; an optional dealer that remains online holds only PRG seeds (\S\ref{sec:opts}).

\paragraph{Online query.} The client encodes its query locally, secret-shares the query \emph{code} to the servers, and the three stages run in sequence. \emph{(1)~\textsf{Filter}}: the servers evaluate the Hamming distance between the shared query code and every shared corpus code under two-server Boolean MPC—a communication-free XOR followed by a small popcount—and select the candidate set $\mathcal{K}$ with a single client-supplied radius threshold, revealing only $\mathcal{K}$ (\S\ref{sec:protocol}). \emph{(2)~\textsf{Rerank}}: each server returns its shares of $|\mathcal{K}|$ candidate embeddings, and the client reconstructs them, scores them against the full-precision query it never released, and keeps the top-$k$—all locally, in plaintext (\S\ref{sec:fetch}). \emph{(3)~\textsf{Fetch}}: the client retrieves the $k$ chosen content blobs by two-server PIR over the replicated ciphertext and decrypts them under its key, so neither server learns which documents were returned (\S\ref{sec:fetch}).

\paragraph{Roadmap.} The rest of the paper develops the learned binary representation (\S\ref{sec:deep-hash}), Boolean Hamming \textsf{Filter} and its cost (\S\ref{sec:protocol}), and client-side \textsf{Rerank} and oblivious \textsf{Fetch} (\S\ref{sec:fetch}). \S\ref{sec:opts} adds two compatible deployment optimizations, and \S\ref{sec:leakage} analyzes leakage.

\subsection{Learned Binary Representation}
\label{sec:deep-hash}

\subsubsection{Towards Binary Representation}
\label{sec:why-binary}

The cost of secure retrieval is first determined by how a document is represented because the representation dictates which secure primitive runs $N$ times per query (except for optimizations we discussed in \S\ref{sec:opts}).

\paragraph{Float and int8 cosine are equally stuck.} A secure cosine over $D$-dimensional vectors is $N\cdot D$ secure multiplications plus a top-$k$ selection. Full precision computation additionally requires fixed-point encoding and truncation.
8-bit quantization available in many libraries~\cite{reimers2019sentence} avoids truncation, but the multiplication count remains unchanged. We implement the latter as a direct full-precision (DFP) baseline (\S\ref{sec:protocol}, \S\ref{sec:eval}): the data owner additive-shares the signed-int8 corpus embeddings over $\Z_{2^{32}}$ during offline setup, the client additive-shares its query in the same domain, and the servers run $N\cdot D$ Beaver multiplications followed by a shared top-$k$ scan. At $N{=}10^6$ this entails $\sim\!7.7\times10^8$ secure multiplications, $\sim\!9.4$ minutes of online wall-clock, and $\sim\!88$\,GB of online traffic per query (Table~\ref{tab:motivation}).

\paragraph{Binary output quantization wastes bits.} The Sentence-transformer library \cite{reimers2019sentence} can binarize an embedding by naively signing each float dimension, which, for a $768$-dimensional encoder, yields a $768$-bit code. This already replaces multiplications with a cheaper Hamming popcount, but it pins the bit budget to the encoder's dimension. The popcount's secure cost is linear in $L$.
We measure the effect of code length directly in Figure~\ref{fig:cost-vs-L}: at $N{=}4096$, the median online wall-clock rises from $6.3$\,ms at $L{=}128$ to $24.1$\,ms at $L{=}768$ ($3.8\times$), the AND-triple count rises $5.6\times$ ($0.57$\,M $\to$ $3.19$\,M), and the online bytes rise $4.7\times$ ($0.35$\,MB $\to$ $1.66$\,MB).
Worse, the dimension-pinned code is not even accurate for its size: signing raw float dimensions is not optimized for the Hamming metric, so it also loses retrieval quality relative to a code learned for that metric, as we show next.
\paragraph{Shorter code with learning.} A deep hash head is a \emph{learned} projection from the encoder's pooled representation to $L$ logits, trained so that Hamming space preserves the top candidates among the original retriever's ranking (\S\ref{sec:hash-arch}). $L$ is a flexible design knob, decoupled from $D$, so we pick the shortest code that preserves sufficient quality. We also observe that a learned projection can concentrate ranking-relevant structure into fewer bits than taking the sign of raw float dimensions: our $128$-bit code outperforms the $768$-bit sentence-transformer binary baseline on three of the four evaluated corpora, with gains of $10.0$--$62.2\%$ (Table~\ref{tab:stbin}). Under XOR sharing, the entire cost is dominated by the popcount over $L$ bits ($\approx\!LN$ ANDs), which the short code minimizes directly.

\begin{figure}[tbp]
\centering
\includegraphics[width=0.78\columnwidth]{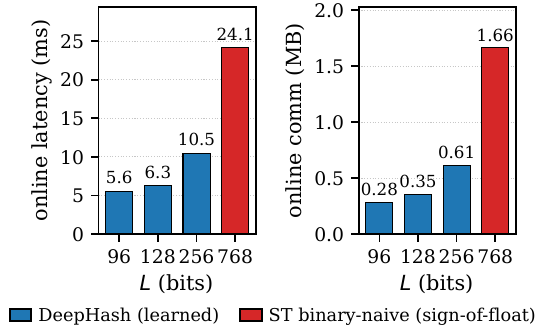}
\vspace{-0.1in}
\caption{Online cost vs.\ code length $L$ at $N{=}4096$ (MP-SPDZ \texttt{semi-bin-party.x}). Both panels grow roughly linearly with $L$, so the $L{=}768$ baseline pays $3.8\times$ the online latency and $4.7\times$ the online traffic of our $L{=}128$ code.}
\label{fig:cost-vs-L}
\vspace{-0.1in}
\end{figure}

\begin{table}[tbp]
\centering
\caption{Hash-only retrieval quality on four BEIR corpora: sentence-transformer $768$-bit binary-naive ($\approx\!768N$ ANDs/query) vs.\ our learned $128$-bit hash.}
\label{tab:stbin}
\vspace{-0.1in}
\small
\setlength{\tabcolsep}{4pt}
\begin{tabular}{l S[table-format=1.4] S[table-format=1.4] r}
\toprule
Dataset & {ST binary ($L{=}768$)} & {hash ($L{=}128$)} & {$\Delta$ (\%)} \\
\midrule
NQ             & 0.3665 & 0.2654 & $-27.6\%$ \\
DBpedia        & 0.2099 & 0.2310 & $+10.0\%$ \\
Climate-FEVER  & 0.0702 & 0.1139 & $+62.2\%$ \\
Webis-Touch\'e & 0.1640 & 0.1768 & $+7.8\%$ \\
\bottomrule
\end{tabular}
\vspace{-0.1in}
\end{table}

\paragraph{Lightweight training recipe.} Producing this code is itself deliberately cheap. We adapt the encoder with LoRA rather than full fine-tuning (\S\ref{sec:hash-arch}), which keeps the training footprint under $48$\,GiB of GPUs memory (two RTX~4090 GPUs) and the wall-clock time under $12$ hours, keeping \sys within reach of clients and model providers with limited hardware. The encoder and head (\S\ref{sec:hash-arch}) and the training objective (\S\ref{sec:objective}) follow; the full training configuration, schedules, and the design alternatives we explored are deferred to Appendix~\suppref{app:hash-training}.
 
\subsubsection{Encoder and Hash Head}
\label{sec:hash-arch}
The dense retriever, or \emph{encoder} $e:\mathcal{X}\to\mathbb{R}^{D}$, maps text (a query/document) to a $D$-dimensional unit vector, with relevance scored by the inner product. Any modern dual-encoder fits this interface; we instantiate $e$ with a pretrained transformer (\texttt{e5-base-v2}, $D{=}768$~\cite{wang2022e5}), so that the float geometry we start from is already strong zero-shot, and we do not attempt to train a retriever from scratch.
 
Our deep hash model is a thin layer placed \emph{on top of} this encoder. A linear hash head $g:\mathbb{R}^{D}\to\mathbb{R}^{L}$ reads the encoder's pooled output and emits $L$ real-valued logits; the binary code is their sign, $b=\sgn(g(e(\cdot)))\in\bits^{L}$. Crucially, the head adds a second, parallel output rather than replacing the first: a single forward pass can produce both the continuous embedding $e(\cdot)$, which the trusted client keeps for the full-precision rerank (\S\ref{sec:protocol}), \textit{and} the $L$-bit code $b(\cdot)$, which is the only object the servers ever compute on (the \textsf{Filter} henceforth). The code length $L$ is a free knob, decoupled from $D$ (\S\ref{sec:why-binary}) and sets the per-document gate count of the \textsf{Filter}. To reshape the encoder for Hamming retrieval without disturbing its zero-shot structure, we adapt it with LoRA~\cite{hu2022lora} rather than full fine-tuning, leaving the bulk of the pretrained weights frozen and lowering memory overhead.
 
\subsubsection{Training Objectives}
\label{sec:objective}
The objective is organized around two concerns—\emph{relevance} (the code must rank the right documents) and \emph{stability} (the float embedding must stay close to its pretrained geometry, so the client's rerank still works). For a training query $q$ with a relevant positive $p$ and a small pool of hard negatives $\{n_i\}_{i=1}^{m}$, the total loss has just three terms,
\[
\mathcal{L} \;=\; \mathcal{L}_{\mathrm{nce}} \;+\; \lambda_{\mathrm{bin}}\,\mathcal{L}_{\mathrm{bin}} \;+\; \lambda_{\mathrm{dist}}\,\mathcal{L}_{\mathrm{dist}},
\]
where, during training, the head's logits are passed through a smooth surrogate $b(\cdot)=\tanh(\beta\,g(e(\cdot)))$ that is differentiable yet approaches the hard $\pm1$ code as $\beta$ grows (Appendix~\suppref{app:hash-training}).
 
\emph{Relevance} is carried by two ranking terms. A contrastive InfoNCE \cite{oord2018infonce} on the continuous embeddings, with temperature $T$ over the in-batch negatives plus the $m$ explicit hard negatives,
\[
\mathcal{L}_{\mathrm{nce}} = -\log \frac{\exp(\cos(e_q,e_p)/T)}{\sum_{d\in\{p\}\cup \mathcal{N}}\exp(\cos(e_q,e_d)/T)},
\]
keeps the float geometry sharp \cite{wang2020uniformity}. To train the hash codes, a \emph{listwise} margin term on the soft codes,
\[
\mathcal{L}_{\mathrm{bin}} = \mathrm{softplus}\Big(\textstyle\log\sum_{i=1}^{m}\exp\big[(\,s_b(q,n_i)-s_b(q,p)\,)/\tau_b\big]\Big),
\]
with $s_b(\cdot,\cdot)$ being the inner product of soft codes. It pushes the positive's code to outrank the entire negative pool (unrelated documents) \emph{jointly}—the ``logsumexp'' is a smooth max over the negatives, which is more discriminative than summing independent pairwise hinges.
 
\emph{Stability} is maintained by a single teacher-distillation term that anchors the adapted encoder to the \emph{frozen} pretrained encoder $e^{\mathrm{T}}$,
\[
\mathcal{L}_{\mathrm{dist}} = \tfrac{1}{|\{q,p,n_i\}|}\!\!\sum_{x\in\{q,p,n_i\}}\!\!\big(1-\cos(e_x, e^{\mathrm{T}}_x)\big),
\]
averaged over the query, the positives, and the negatives. Without this anchor, the binary loss is free to drag the encoder into a geometry that ranks well in Hamming space but reranks poorly in float.
 
\subsubsection{A Minimal Training Recipe}
\label{sec:hash-minimal}
The deep-hashing literature, grown largely around image retrieval, surrounds a ranking loss like $\mathcal{L}_{\mathrm{bin}}$ with a battery of \emph{code-quality regularizers}: a quantization penalty that forces logits to the $\pm1$ corners~\cite{li2015feature,zhu2016deep}, a bit-balance penalty that keeps each bit firing roughly half the time, and a bit independence (decorrelation) penalty that discourages redundant bits~\cite{do2016bdnn}. These desiderata descend from classical learning-to-hash and are cataloged across recent surveys~\cite{wang2014hashing,luo2023survey,he2025survey}. Each adds a loss weight, and several add a schedule to a pipeline that is already delicate to train.

We use \emph{none} of them. Under our relevance objective and the teacher anchor, the codes are already balanced and the logits saturate on their own, an effect also obtainable without an explicit balance term~\cite{hoe2021one,li2021bihalf,shen2018nash}—so adding the penalties buys no measurable quality and only enlarges the hyperparameter search. Dropping them makes the recipe both simpler and, in our setting, stronger: at matched float quality, the bare objective produces \emph{more} discriminative codes. We formulate each regularizer and explain the mechanism by which the active objective already supplies its effect in Appendix~\suppref{app:hash-training}.

Since \emph{the code length is a gate budget}, (\S\ref{sec:stage1}), one should pick the shortest code that preserves quality. We train the model with 96, 128, and 256 output code bits. On this encoder, quality saturates by $L{=}128$ and only marginally improves at $L{=}256$ (\S\ref{sec:eval-quality}), and we adopt that as the default.

\subsection{Coarse Filtering under Two-Server MPC}
\label{sec:protocol}
The \textsf{Filter} scores every document against the query once under two-server Boolean MPC and reveals a coarse candidate set $\mathcal{K}$. It has two parts—a free-XOR-plus-popcount Hamming distance (\S\ref{sec:stage1}) and a candidate-selection step we reduce to a single comparison and a single reveal (\S\ref{sec:stage2}).

\subsubsection{Hamming Distance via Wallace Popcount}
\label{sec:stage1}
Each server holds XOR shares $\share{q_j}_A,\share{q_j}_B$ of every query bit and $\share{H_{ij}}_A,\share{H_{ij}}_B$ of every corpus-code bit. For document $i$ and bit $j$, server $P_X$ locally computes $\share{x_{ij}}_X=\share{q_j}_X\oplus\share{H_{ij}}_X$. The two results satisfy $\share{x_{ij}}_A\oplus\share{x_{ij}}_B=q_j\oplus H_{ij}$, so they share a bit that is one exactly when the query and document differ at position $j$.

The servers sum these $L$ shared indicator bits with a Wallace-tree carry-save popcount~\cite{wallace1964suggestion}. At each binary weight, a $3{:}2$ compressor replaces three shared bits $a,b,c$ with a sum bit $s=a\oplus b\oplus c$ at the same weight and a carry bit $u=\mathrm{maj}(a,b,c)=c\oplus\big((a{\oplus}c)\wedge(b{\oplus}c)\big)$ at the next weight. The identity $a+b+c=s+2u$ preserves the represented integer at every layer. Repeating the compression leaves two shared bit vectors; a Boolean carry-propagate addition produces shares of $d_i=\sum_{j=1}^{L}x_{ij}=\mathrm{HW}(q\oplus H_i)$.

\paragraph{Circuit cost.} The difference bits and each compressor's sum bit use local XORs. Each carry bit uses one secure AND and one Beaver triple through the majority expression above. The complete popcount has $O(\log L)$ depth and consumes approximately $L$ AND gates per document, or $LN$ per query. It is the dominant corpus-linear component of the fixed-radius circuit, so shortening the learned code directly reduces secure filtering cost (\S\ref{sec:why-binary}).

\subsubsection{Candidate Selection with a Calibrated Radius}
\label{sec:stage2}
Given the shared distances, the \textsf{Filter} must reveal enough candidates for the client's final rerank to preserve retrieval quality. Selecting the exact top-$K$ would require an oblivious sort or selection over the shares~\cite{hamada2012sorting,bogdanov2014sorting}, which is costly in the cryptographic domain. A \emph{radius threshold} that reveals every document whose shared distance falls below a radius $t$ is far cheaper but produces a query-dependent candidate count and requires a concrete public radius.

As a baseline, with direct full-precision (DFP) embeddings, $\Pi_{\mathrm{DFP}}$ scores additive-shared int8 embeddings with $N\!\cdot\!D$ Beaver multiplications, followed by a shared top-$k$ scan. Batching all independent products reduces the inner product to two online rounds, but does not change its corpus-linear arithmetic work or traffic. This is the wall shown in Table~\ref{tab:motivation}, and it's not practically tractable beyond $\sim\!10^5$ documents.

The second protocol, $\Pi_{\mathrm{BS}}$, keeps the codes but reads the radius from the data, binary-searching the shared distances for the smallest $t^\star$ whose cumulative count reaches $K{=}\lceil\rho N\rceil$. It is gate-cheap but flawed in two ways: it is \emph{interactive}, spending $\lceil\log_2(L{+}1)\rceil\!\approx\!8$ reveal rounds per query. On any link with non-trivial round-trip time, the cost of these rounds dominates the protocol latency. We discuss more details and show the full forms of both protocols in Appendix~\suppref{app:strawmen}.

\paragraph{\sys: calibrated fixed-radius.} $\Pi_{\mathrm{FR}}$ chooses one public radius for each (checkpoint, corpus) pair before deployment. The calibration data are a small labeled sample from the target corpus's query distribution. In our evaluation, we partition each BEIR corpus's official test queries: for each of five seeds, we divide the ordered query list into equal strata and sample one query per stratum, giving 100 calibration queries (ten for Touch\'e), while the remaining queries form that seed's held-out evaluation set. The checkpoint is trained on MS MARCO; these BEIR queries are used only to choose the radius. For a target retention $\eta$, defined as the float-reranked NDCG divided by the same checkpoint's full-corpus float NDCG, each split selects the smallest integer radius that returns at least $k$ candidates for every calibration query and reaches retention $\eta$. The deployed radius $\hat t$ is the median of the five proposals (Algorithm~\ref{alg:calib}).

\paragraph{Online selection.} For every shared distance $d_i$, the servers evaluate the shared indicator $m_i=[d_i\le\hat t]$ against the public radius. All $N$ comparisons run in parallel and cost approximately $16N$ ANDs; the base protocol then reveals the indicator vector $m$ once, yielding $\mathcal{K}=\{i:m_i=1\}$. Combining the popcount and threshold, $\Pi_{\mathrm{FR}}$ uses approximately $[L+2\lceil\log_2(L{+}1)\rceil]N$ ANDs per query; its gate count and communication are linear in $N$. Different queries can produce different candidate counts under the same radius, so the evaluation reports their median and tail. The hidden-padding variant in Appendix~\suppref{app:leakage_defense} sends the two output shares only to the client and reveals only a client-padded union to the servers; it adds one client broadcast without changing the filter circuit.

\begin{algorithm}[t]
\caption{Offline NDCG-to-radius calibration}
\label{alg:calib}
\begin{algorithmic}[1]
\Require stratified calibration splits $S_1,\ldots,S_5$; codes $H$; float embeddings $E$; final rank $k$; NDCG retention $\eta$
\For{$j=1,\ldots,5$}
  \State $F_j\gets\operatorname{NDCG@}k(\operatorname{FloatTop}k(S_j,E))$
  \For{$t=0,\ldots,L$}
    \State $\mathcal K_q(t)\gets\{d:\operatorname{Hamming}(q,d)\le t\}$ for $q\in S_j$
    \State $R_j(t)\gets\operatorname{NDCG@}k(\operatorname{FloatRerank}(\mathcal K_q(t),E))/F_j$
  \EndFor
  \State $t_j\gets\min\{t:R_j(t)\ge\eta\ \land\ \min_{q\in S_j}|\mathcal K_q(t)|\ge k\}$
\EndFor
\State \Return $\hat t=\operatorname{median}(t_1,\ldots,t_5)$
\end{algorithmic}
\end{algorithm}

\subsubsection{Optional Filtering Optimizations}
\label{sec:opts}

The base system needs only the two clouds: they generate buffered triples with Silent OT (\S\ref{sec:mpc-bg}) and scan every code. Two optional but beneficial optimizations target different deployment bottlenecks without changing the fixed-radius circuit.
\textbf{Private cluster pruning} lowers online work when a deployment can trade a small retrieval quality for higher throughput; a \textbf{seeded institutional dealer} moves query-independent triple generation from the clouds to a small owner-operated service.
Either can be enabled alone, and their effects directly compose because one reduces triple consumption while the other accelerates triple supply.

\paragraph{Private cluster pruning.}
\label{sec:private-prune}

Private pruning limits MPC to a fixed number of padded Hamming clusters while hiding which clusters the query selects. With $C$ clusters and capacity factor $\alpha$, each bucket has $\lceil\alpha N/C\rceil$ slots, so probing a public number $p$ fixes the scan at $M=p\lceil\alpha N/C\rceil$ rows. During setup, the owner runs capacity-constrained binary $k$-means, pads every cluster to this capacity, masks the resulting bucket database with a seeded PRG stream, and replicates the masked database at both clouds. The client retains only centroids and mask seeds. For a query, it selects the $p$ nearest centroids locally, retrieves the corresponding buckets into XOR shares with two-server PIR, and feeds those shares directly into the Boolean registers that run $\Pi_{\mathrm{FR}}$ (see Algorithm~\ref{alg:private-prune}). Each cloud only sees uniformly random PIR selectors and mask corrections, plus public $p$ and bucket capacity; it learns neither the selected cluster IDs nor a query-dependent scan size.

\begin{algorithm}[t]
\caption{Private fixed-volume cluster pruning}
\label{alg:private-prune}
\begin{algorithmic}[1]
\Require query code $q$; centroids $Z$; masked padded buckets $\widetilde B$; public probes $p$ and radius $\hat t$
\State $J\gets$ indices of the $p$ nearest centroids to $q$
\For{$j\in J$}
  \State Client XOR-shares one-hot selector $e_j$ as $(r_j^A,r_j^B)$
  \State $P_X$ computes $s_j^X\gets\operatorname{PIR}(\widetilde B,r_j^X)$ for $X\in\{A,B\}$
  \State Client sends fresh XOR shares of bucket $j$'s mask
  \State $P_A,P_B$ correct $(s_j^A,s_j^B)$ into shares of padded bucket $B_j$
\EndFor
\State $P_A,P_B$ run $\Pi_{\mathrm{FR}}(q,\bigcup_{j\in J}B_j,\hat t)$
\State Client removes dummy rows and float-reranks the revealed candidates
\end{algorithmic}
\end{algorithm}

The public probe count is a quality-speed knob: increasing $p$ approaches the full scan while computing on more clusters. We freeze one configuration across all evaluated corpora using only calibration-query containment of the full-FR float top 10, then evaluate it once on held-out queries (\S\ref{sec:eval-throughput}).

\paragraph{Seeded institutional dealer.}
\label{sec:dealer}

An organization that can operate a small trusted server can use it as a seeded dealer for Boolean triples, following an established (2+1)-party MPC model~\cite{obsidian26, riazi2018chameleon, decock2021highperformance}. Cloud $P_A$ expands $(a_A,b_A,c_A)$ from its private seed, $P_B$ expands $(a_B,b_B)$ from another, and the dealer expands both streams, computes $c_B=(a_A\oplus a_B)\land(b_A\oplus b_B)\oplus c_A$, and sends only this one-bit-per-triple correction to $P_B$. The service stores no corpus, query, embedding, or document content; it performs sequential PRG expansion that is \textit{far smaller} than storing and scanning the outsourced retrieval index.

The dealer fills the same triple buffer as cloud Silent OT, so availability changes performance. If the dealer is absent or temporarily unavailable, the two clouds refill the buffer with Silent OT and continue along the base path. Private pruning is compatible: it simply reduces how many triples either source must provide.

\subsection{Rerank and Oblivious Fetch}
\label{sec:fetch}
Once the \textsf{Filter} reveals $\mathcal{K}$, the precise work runs on the trusted client over this small candidate set.

\paragraph{Offline content storage.}
\label{sec:offline-index}
The codes $H$ and the full embeddings are XOR-shared between the two servers in the offline setup (\S\ref{sec:overview}). Fix a public ciphertext-row width $B_{\mathrm{ct}}$ at deployment and let $B_{\mathrm{pt}}$ subtract the fixed nonce and authentication-tag overhead. The owner length-prefixes and pads every serialized document to exactly $B_{\mathrm{pt}}$ bytes, then encrypts each row with AES-256-GCM under a client-held key and a distinct nonce, producing exactly $B_{\mathrm{ct}}$ stored bytes. The same ciphertext rows are replicated at both servers. The deployment chooses $B_{\mathrm{pt}}$ at least as large as its maximum indexed record; larger application objects are segmented before corpus construction. Replication enables two-server PIR (\S\ref{sec:mpc-bg}) without reconstructing plaintext on either server; fixed-width padding removes per-record byte length from setup leakage.

The owner samples a uniform permutation $\pi\leftarrow S_N$ independently of the codes and applies it consistently to code shares, embedding shares, and ciphertext rows. The client retains the logical-ID--to--slot map, while each server sees only the permuted physical slots. This data-independent layout prevents physical adjacency from revealing hash-prefix or cluster membership. Private pruning stores each padded bucket as a PIR record; the record index is hidden by PIR and the record contents remain secret-shared after retrieval.

\paragraph{Client rerank.} Each server sends the client its XOR shares of the $|\mathcal{K}|$ candidate embeddings; the client reconstructs them, reranks on the full-precision query embedding it never released, and selects the top-$k$. This is a $|\mathcal{K}|\!\times\!D$ float dot product costing only microseconds on the client.

\paragraph{Oblivious fetch.} The client then fetches the $k$ chosen content blobs. Because the content is replicated client-encrypted ciphertext, the fetch touches no secure computation and only the key-holding client decrypts; each server sees only ciphertext (IND-CPA).
\begin{itemize}[leftmargin=0.12in]\setlength\itemsep{0.2em}
    \item \emph{Design A (download-all)} pulls all $|\mathcal{K}|$ ciphertexts and decrypts the top-$k$ locally; the choice is hidden because the client never echoes it, but the download grows with $|\mathcal{K}|$. 
    \item  \emph{Design B (PIR-over-$\mathcal{K}$, \underline{default})} runs $k$ classical two-server PIR queries~\cite{chor1998private} over the replicated ciphertext—local masked XOR-reduce at each server, zero secure computation, one round—so the client downloads only the $k$ chosen blobs and neither server learns which $k$ of $|\mathcal{K}|$ were picked.
\end{itemize}

The two cross at $k^\star{=}|\mathcal{K}|\,B_{\mathrm{ct}}/\big(2(B_{\mathrm{ct}}{+}\lceil|\mathcal{K}|/8\rceil)\big)$ items fetched (Figure~\ref{fig:fetch}), where $B_{\mathrm{ct}}$ is the ciphertext row width: in a typical RAG regime of $k\!\sim\!10$, Design~B's flat-in-$|\mathcal{K}|$ download wins by roughly two orders of magnitude, while above $\sim\!40\%$ of $|\mathcal{K}|$ download-all is cheaper. Either way, the servers learn only the candidate set; the precise ranking and the content plaintext stay on the trusted client.

\begin{figure}[t]
\centering
\includegraphics[width=0.85\columnwidth]{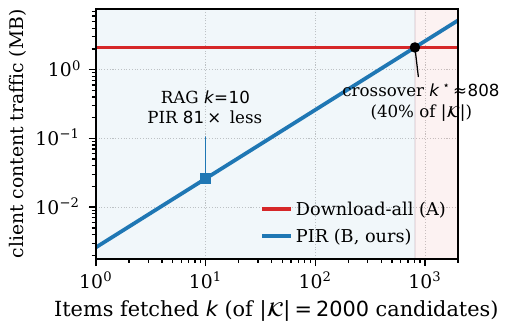}
\vspace{-0.1in}
\caption{Oblivious content fetch: client content traffic vs.\ items fetched $k$, at $|\mathcal{K}|{=}2000$ candidates and a $1$\,KB content blob (the shared $\approx\!3$\,MB candidate-embedding download excluded). \emph{Download-all} pulls all $|\mathcal{K}|$ ciphertext rows; \emph{PIR} fetches only the $k$ chosen blobs (linear).}
\label{fig:fetch}
\vspace{-0.1in}
\end{figure}

\section{Leakage Analysis}
\label{sec:leakage}
\paragraph{The servers learn a precisely defined access pattern.} \sys samples a data-independent secret slot permutation at setup and pads every content row to a public width. The resulting setup leakage consists of public dimensions, fixed row width, and protocol parameters. For a full fixed-radius scan, the per-query leakage is the public radius and the revealed set of permuted physical slots,
\[
\mathcal{L}_{\mathrm{full}}(q)=\big(\hat t,\mathcal{K}_q\big).
\]
Private pruning has a different profile: the bucket IDs remain PIR-hidden, and a server learns only the fixed scan dimensions and the indicator vector over the transient, freshly shared bucket buffer. Codes, embeddings, query bits, selected cluster IDs, the within-candidate ranking, and fetched content remain hidden from either server. Appendix~\suppref{app:leakage} formalizes setup and online leakage for both variants and proves adaptive multi-query simulation security from the security of the underlying 2PC, preprocessing, encryption, and PIR components~\cite{goldreich1987gmw,lindell2017simulate}.

\paragraph{Leakage quantification.} Repeated full-scan candidate sets can reveal approximate unlabeled document neighborhoods even though the protocol never opens codes, embeddings, ranking, content, or queries. Ordered-domain and volume-based reconstruction attacks~\cite{kellaris2016generic,grubbs2018pump} do not directly apply: our search has no ordered plaintext domain, and fixed-width content rows suppress per-document byte length.

\begin{table*}[tbp]
\centering
\caption{Retrieval quality and candidate cost on four BEIR corpora. ``Float'' is the full-corpus float reference, ``Hash'' the no-rerank floor, ``BS@1K'' exact Hamming top-1000 followed by float reranking, and ``FR'' our fixed radius calibrated for 95\% final-NDCG retention. $R_{\mathrm{float}@10}$ is the fraction of the float top-10 present in FR's candidates; $K_{50}/K_{95}$ are the median and 95th-percentile candidate counts. FR, recall, and candidate counts report mean $\pm$ standard deviation over five stratified calibration partitions (ten calibration queries for Webis-Touch\'e and 100 otherwise).}
\label{tab:quality}
\vspace{-0.1in}
\small
\setlength{\tabcolsep}{8pt}
\begin{tabular}{l c c c c c c r}
\toprule
Dataset & $L$ & Float & Hash & BS@1K & FR NDCG@10 & $R_{\mathrm{float}@10}$ & $K_{50}/K_{95}$ \\
\midrule
\multirow{3}{*}{NQ}
 & 96 & 0.5376 & 0.2329 & 0.5114 & 0.5261 $\pm$ 0.0010 & 0.938 $\pm$ 0.000 & 4,779 / 15,536 \\
 & 128 & 0.5363 & 0.2729 & 0.5214 & 0.5250 $\pm$ 0.0011 & 0.931 $\pm$ 0.000 & 1,952 / 8,191 \\
 & 256 & 0.5374 & 0.3385 & 0.5308 & 0.5298 $\pm$ 0.0010 & 0.953 $\pm$ 0.000 & 1,050 / 4,571 \\
\midrule
\multirow{3}{*}{DBpedia}
 & 96 & 0.3986 & 0.1941 & 0.3740 & 0.3794 $\pm$ 0.0032 & 0.845 $\pm$ 0.005 & 3,373 / 12,139 \\
 & 128 & 0.3996 & 0.2332 & 0.3857 & 0.3865 $\pm$ 0.0044 & 0.831 $\pm$ 0.007 & 1,207 / 5,749 \\
 & 256 & 0.4002 & 0.2987 & 0.3958 & 0.3944 $\pm$ 0.0024 & 0.869 $\pm$ 0.006 & 500 / 2,771 \\
\midrule
\multirow{3}{*}{Climate-FEVER}
 & 96 & 0.2570 & 0.0973 & 0.2516 & 0.2524 $\pm$ 0.0021 & 0.706 $\pm$ 0.001 & 1,056 / 2,633 \\
 & 128 & 0.2597 & 0.1150 & 0.2597 & 0.2539 $\pm$ 0.0020 & 0.672 $\pm$ 0.001 & 382 / 1,055 \\
 & 256 & 0.2625 & 0.1471 & 0.2653 & 0.2624 $\pm$ 0.0024 & 0.748 $\pm$ 0.001 & 377 / 1,001 \\
\midrule
\multirow{3}{*}{Webis-Touch\'e}
 & 96 & 0.2681 & 0.1678 & 0.2733 & 0.2568 $\pm$ 0.0208 & 0.903 $\pm$ 0.013 & 315 / 4,450 \\
 & 128 & 0.2686 & 0.1807 & 0.2706 & 0.2549 $\pm$ 0.0208 & 0.947 $\pm$ 0.012 & 385 / 5,289 \\
 & 256 & 0.2691 & 0.2129 & 0.2681 & 0.2546 $\pm$ 0.0207 & 0.886 $\pm$ 0.015 & 128 / 3,254 \\
\bottomrule

\end{tabular}
\vspace{-0.2in}
\end{table*}

We quantify this signal with a normalized co-occurrence estimator over ciphertext slots, inspired by access-pattern inference~\cite{islam2012access,cash2015leakage}. It predicts edges between stable encrypted slot identifiers; it does not perform graph alignment or attach plaintext labels. At $L{=}128$, conditional precision@10 is $0.017$--$0.070$ under held-out workloads and $0.099$--$0.280$ under a 20k-query coverage stress. A client-side hidden-padding variant reduces the stress result by $39$--$57\%$ at a padding ratio $r{=}2$. Appendix~\suppref{app:leakage} defines the estimator, its evaluation universe, and the modified protocol that reveals only the padded union.

\section{Evaluation}
\label{sec:eval}

\subsection{Implementation and Setup}
\label{sec:impl}

We implement FR and binary search in MP-SPDZ's semi-honest Boolean engine~\cite{keller2020mpspdz}. The corpus is bit-sliced: each code position is an \texttt{sbit} vector with one SIMD lane per document, and each query bit is broadcast across the lanes. Both protocols use our single-AND-majority Wallace popcount and differ only in candidate selection. Calibration fixes the public radius before deployment, so FR compiles one schedule per $(N,L,\hat t)$. DFP runs in the arithmetic engine over $\Z_{2^{32}}$ with additive shares prepared during setup, batched inner products, and repeated shared argmax.

Our preprocessing port converts libOTe Silent OTs~\cite{libOTe,boyle2019silentot} into MP-SPDZ's packed Boolean triples. Private pruning retrieves padded buckets through long-lived two-server PIR endpoints and loads the resulting XOR shares directly into Boolean registers. Content is stored as fixed-width AES-256-GCM ciphertext and fetched with the same two-server PIR kernel. The seeded dealer expands five AES-128-CTR streams and emits the correction share for each triple.

\paragraph{Setup.} We train the deep hash model on MS MARCO and evaluate zero-shot on four BEIR corpora with 383K to 5.42M documents: NQ~\cite{kwiatkowski2019natural}, DBpedia~\cite{hasibi2017dbpedia}, Climate-FEVER~\cite{diggelmann2020climatefever}, and Touch\'e~\cite{bondarenko2020touche}, reporting NDCG@10~\cite{thakur2021beir}. The encoder is \texttt{e5-base-v2} with LoRA, at $L\in\{96,128,256\}$. For each of the five partition seeds, calibration uses 100 queries (ten for Touch\'e, which has 49 test queries), and evaluation uses all remaining queries; the deployed radius is the median proposal. The pruning configuration is frozen across corpora: 256 clusters, a capacity factor of 1.2, and 43 probes, giving a roughly 20\% scan of each corpus. 

Protocol measurements run both parties on a dual-socket AMD EPYC~9654 host over TCP using MP-SPDZ~\cite{keller2020mpspdz} and libOTe~\cite{libOTe}. The online \textsf{Filter} uses one core per party; cloud preprocessing runs 48 parallel Silent-OT worker pairs, and the seeded dealer uses one core. To test performance in varied network environments, we throttle the TCP bandwidth with Linux \texttt{tc} to 100\,Mbps, 1\,Gbps, or 10\,Gbps for cross-cloud path and each client/dealer-to-cloud link. We study retrieval quality, protocol cost, online latency, and the communication and sustained throughput with the optional optimizations.

\begin{table*}[t]
\centering
\caption{Online candidate-generation cost at $L{=}128$. Webis-Touch\'e is measured directly; the three larger rows scale the audited per-document rates, independently validated at $10^6$ and $10^7$. DFP scales from its measured $N{=}4096$ anchor.}
\label{tab:protocols}
\vspace{-0.1in}
\small
\setlength{\tabcolsep}{7pt}
\begin{tabular}{l r S[table-format=7.0] S[table-format=6.0] S[table-format=4.1] S[table-format=3.2] S[table-format=4.1] S[table-format=3.2]}
\toprule
& & \multicolumn{2}{c}{DFP (no hash)} & \multicolumn{2}{c}{BS (binary search)} & \multicolumn{2}{c}{FR (calibrated, ours)} \\
\cmidrule(lr){3-4}\cmidrule(lr){5-6}\cmidrule(lr){7-8}
Dataset & {$N$} & {lat.\ (ms)} & {comm (MB)} & {lat.\ (ms)} & {comm (MB)} & {lat.\ (ms)} & {comm (MB)} \\
\midrule
NQ              & \num{2681468} & 1507246 & 235601 & 3802.5 & 332.51 & \textbf{1284.8} & \textbf{230.61} \\
DBpedia         & \num{4635922} & 2605839 & 407325 & 6574.0 & 574.87 & \textbf{2221.3} & \textbf{398.70} \\
Climate-FEVER   & \num{5416593} & 3044652 & 475917 & 7681.1 & 671.68 & \textbf{2595.4} & \textbf{465.83} \\
Webis-Touch\'e  & \num{382545}  & 215027  & 33611  & 542.5  & 47.44  & \textbf{183.3}  & \textbf{32.90} \\
\bottomrule
\end{tabular}
\vspace{-0.05in}
\end{table*}

\begin{figure*}[thbp]
\centering
\begin{minipage}[thbp]{0.7\textwidth}
\centering
\includegraphics[width=0.85\textwidth]{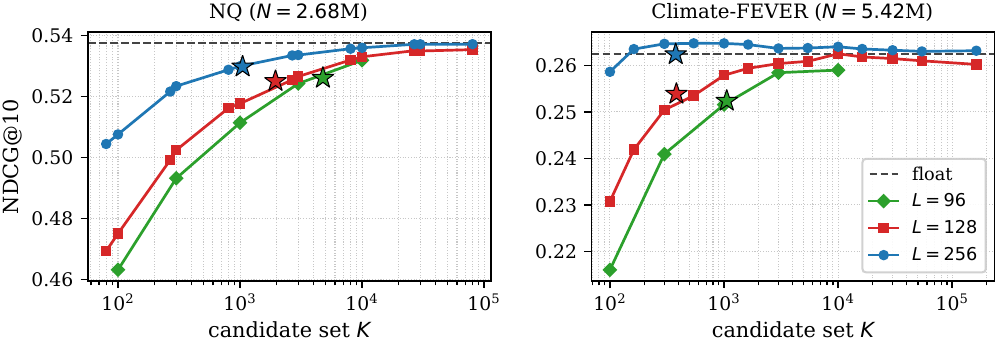}
\vspace{-0.15in}
\caption{Hybrid NDCG@10 vs.\ candidate count $K$. The dashed line is the full-corpus float reference; stars mark FR at its mean median candidate count. NQ shows the smooth width tradeoff: median candidates fall from 4,779 at 96 bits to 1,952 at 128 and 1,050 at 256; Climate-FEVER reaches the calibrated target with 1,056/382/377 candidates.}
\label{fig:quality}
\end{minipage}
\hfill
\begin{minipage}[thbp]{0.28\textwidth}
\centering
\includegraphics[width=0.85\textwidth]{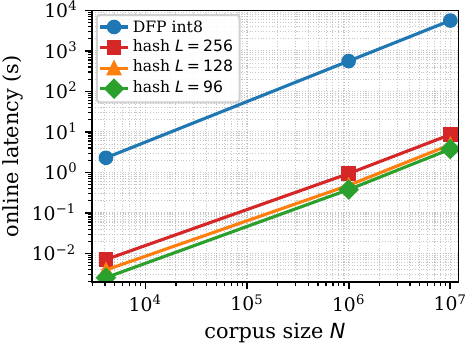}
\vspace{-0.05in}
\caption{Online latency vs.\ corpus size. Hash through $10^6$ are measured in MP-SPDZ. DFP is scaled from its measured $N{=}4096$ anchor.}
\label{fig:latency}
\end{minipage}
\vspace{-0.1in}
\end{figure*}

\subsection{Retrieval Quality}
\label{sec:eval-quality}
Table~\ref{tab:quality} reports quality and candidate counts across three separately trained code widths. We found the following. 

\ding{172} FR retains a high $95.1$--$98.1\%$ of full-corpus float NDCG at $L{=}96$, $95.2$--$97.8\%$ at $L{=}128$, and $95.3$--$99.8\%$ at $L{=}256$. This suggests that 96--128 bits could effectively preserve the high-ranking candidates with hashing, with higher length leading to diminishing returns. We also tried training for 64 bits, but the performance had major degradation compared to 96 bits, forcing a significantly larger candidate set to retain the NDCG. 
In particular, the 96-bit medians are 315--4,779, compared to the 128-bit range of 382--1,952 and the 256-bit range of 128--1,050; while on the difficult NQ and DBpedia corpora, moving from 128 to 96 bits can enlarge the median pool by $2.4\times$ and $2.8\times$. But the shorter code also reduces measured online filtering from 183 to 146\,ms at Webis-Touch\'e and from 482 to 378\,ms at one million documents. 

\ding{173} At $L{=}128$, float top-10 candidate recall ranges from 0.672 to 0.947, while final NDCG retention stays above 0.952; this is because the float top-10 is \textit{not} a golden set of relevant documents but rather a reflection of the behavioral similarity between the hash model and the original one.
Reporting only candidate recall would therefore misstate application quality, while reporting only NDCG would hide how faithfully the candidate generator reproduces the float ranking. We report both for comprehensiveness.

\ding{174} The fixed-$K$ comparison exposes the budget trade. Against BS@1K, FR spends more candidates on NQ and DBpedia to improve NDCG, but fewer on Climate-FEVER and Webis-Touch\'e under the calibrated 5\% quality relaxation. Figure~\ref{fig:quality} shows the same operating points over the fixed-$K$ curves.

\ding{175} The hash-only floor remains far below the reranked result, especially at 96 bits, confirming that the compact code is a candidate proposer and cannot function directly as the final ranker.


\subsection{Protocol Overhead}
\label{sec:eval-protocols}
Table~\ref{tab:protocols} measures the three candidate-generation protocols (DFP, BS, FR; \S\ref{sec:protocol} and Appendix~\suppref{app:strawmen}) at $L{=}128$ across four BEIR corpora for online latency and communication per query. We first analyze the results across two axes central to our design: the \emph{representation} (DFP vs.\ hash) and the \emph{protocol} (BS vs.\ FR), then analyze the cost of the rerank and fetch stage.

\paragraph{Representation: FR vs.\ DFP} The first and last column groups isolate the representation. DFP scores every document with an int8 cosine and pays from $215$\,s at Webis-Touch\'e to $3{,}045$\,s at Climate-FEVER, moving $33.6$--$475.9$\,GB. FR replaces the $N\!\cdot\!D$ arithmetic multiplications with a free XOR and a $128$-bit popcount, answering in $183$\,ms--$2.60$\,s with $32.9$--$465.8$\,MB. This is a $1{,}173\times$ latency reduction and a $\sim\!1{,}022\times$ communication reduction.

\paragraph{Protocol: FR vs.\ BS} The last two column groups isolate candidate selection: BS and FR run the identical popcount, but BS performs eight data-dependent count reveals, while FR applies one pre-calibrated threshold and reveals one indicator vector (\S\ref{sec:stage2}). Consequently, BS is $2.9$--$3.0\times$ slower and moves $1.4\times$ more data. For example, on Webis, it uses a total of 531 MPC rounds against FR's 27.


\paragraph{Rerank and Fetch bandwidth.} Once the \textsf{Filter} reveals $\mathcal{K}$, each server ships its shares of the $|\mathcal{K}|$ candidate embeddings ($2|\mathcal{K}|D$ bytes) and the client reranks locally with a $|\mathcal{K}|\!\times\!D$ float dot product. At $L{=}128$, the median candidate pools add 0.59--3.00\,MB across the four corpora, below 2\% of the corresponding 32.9--465.8\,MB \textsf{Filter} traffic. Even at $K_{95}$, the rerank stage stays below 6\% on the three million-scale corpora.
\textsf{Fetch}, with two-server PIR over replicated ciphertext, is cheaper still: each server XOR-reduces over the candidate rows with no secure computation, and at $k{=}10$ the client downloads only its chosen blobs, two orders below download-all and far from bottlenecking the pipeline.

\paragraph{Scaling behavior} Figure~\ref{fig:latency} extends the measurements to million- and ten-million-document corpora. At $L{=}128$, FR takes $0.482$\,s (86\,MB) for $N{=}10^6$ and $4.85$\,s (860\,MB) for $N{=}10^7$; DFP would take $9.4$ minutes and $1.56$ hours, respectively. The FR-over-DFP gap \emph{widens} as $L$ shrinks.

\subsection{Online Latency and Scaling}
\label{sec:eval-latency}

\begin{table*}[t]
\centering
\caption{Online latency at 10\,Gbps in s/query; parentheses give baseline$/\sys$-(Full FR) latency. $\dagger$: For Pisces, TO denotes a 300-s timeout cutoff after proper candidate-only cache warm-up.}
\label{tab:e2e}
\vspace{-0.1in}
\small
\renewcommand{\arraystretch}{1}
\begin{tabular*}{\textwidth}{@{\extracolsep{\fill}} l c c c c c}
\toprule
& & \multicolumn{4}{c}{Online Latency (s/query)} \\
\cmidrule(lr){3-6}
System & Method & {NQ} & {DBpedia} & {Climate-FEVER} & {Webis-Touch\'e} \\
\midrule
RemoteRAG~\cite{chang2025remoterag} & PHE & 9.794 (6.7$\times$) & 12.627 (5.0$\times$) & 14.284 (4.8$\times$) & 5.211 (24.8$\times$) \\
$p^2$RAG~\cite{ming2026p2rag} & 2PC, FSS & 10.287 (7.0$\times$) & 18.492 (7.3$\times$) & 21.668 (7.3$\times$) & 1.408 (6.7$\times$) \\
PANTHER~\cite{li2025panther} & PIR+MPC & \multicolumn{3}{c}{\textit{OOM at 1M documents}} & 18.387 (87.5$\times$) \\
Pisces$^\dagger$~\cite{liang2026pisces} & PSI+MPC & 182.6 (124.1$\times$) & TO ($>300$; $>118.0\times$) & TO ($>300$; $>101.1\times$) & 23.747 (113.0$\times$) \\
\midrule
\multirow{2}{*}{\sys\ (Ours)} & Full FR & 1.471 & 2.542 & 2.968 & 0.210 \\
& Private Pruning+FR & \bfseries 0.441 & \bfseries 0.872 & \bfseries 1.090 & \bfseries 0.061 \\
\bottomrule
\end{tabular*}
\renewcommand{\arraystretch}{1}
\vspace{-0.05in}
\end{table*}

\paragraph{Comparison vs.\ cryptographic-RAG baselines.} Table~\ref{tab:e2e} compares five systems, including ours, using the same E5-base-v2 embeddings on the same host machine. All systems assume a long-lived service: one-time key generation, index construction, protocol preprocessing, and process startup are outside the measured query path.

RemoteRAG~\cite{chang2025remoterag} searches its plaintext index with a differentially private perturbed query, then uses partially homomorphic encryption (PHE) to rerank. We use its paper-studied $r{=}0.05$ setting ($\varepsilon{=}15360$ at 768 dimensions), a 1024-bit Paillier key, and 96 workers; it takes a total latency of 5.21--14.28\,s across the four corpora.
$p^2$RAG~\cite{ming2026p2rag} secret-shares the full 768-dimensional embeddings to two non-colluding servers like us and securely scores every document before selecting the top results. With 192 threads and its published communication modeled at 10\,Gbps and 0.1\,ms RTT, this full-embedding scan takes 1.41--21.67\,s across the four corpora.
PANTHER~\cite{li2025panther} privately retrieves clustered posting lists and then scores their contents with MPC. To hide which list length was selected, its artifact pads the lists in each group to a common maximum; the resulting PIR database exhausts 256\,GB on all three million-scale corpora, while Webis-Touch\'e reaches $>99\%$ float-top-10 agreement in 18.39\,s.
Pisces~\cite{liang2026pisces}'s official SimHash rule stops at $15\%$ of the corpus before neighborhood expansion. Our patched complete Webis-Touch\'e run returns 58.0K--113.5K candidates, retains $89.39\%$ of float top-10 and $95.85\%$ of NDCG@10, and takes 23.75\,s at 10\,Gbps. On NQ, a steady-state query with 575.7K candidates takes 168.1\,s on loopback and 182.6\,s after 10-Gbps serialization. The corresponding DBpedia and Climate-FEVER queries select 1.09M and 1.20M candidates and both exceed a 300-s complete-path cutoff after a candidate-only cache warm-up ($\dagger$).

The dominant difference is how much data each system processes under expensive cryptography. The full \sys scan is $4.8$--$6.7\times$ faster than the closest completed baseline on each corpus. Private Cluster Pruning reduces the fixed MPC scan to about 20\% of the corpus and lowers latency to 61--1,090\,ms, a $13.1$--$22.9\times$ advantage over the closest baseline while retaining $93.9$--$97.3\%$ of full-float NDCG.

\subsection{Optional Optimizations Across Deployments}
\label{sec:eval-throughput}

Table~\ref{tab:optional} evaluates the two optimizations in \S\ref{sec:opts} independently and combined. \emph{Full} (F) scans every code; \emph{pruned} (P) uses the universal private cluster pruning configuration. \emph{Cloud OT} (O) uses 48 parallel Silent-OT worker pairs; \emph{dealer} (D) uses one owner-side CPU core. 

\paragraph{Pruning trades minor quality loss for lower online cost.} We select one clustering configuration using calibration queries only. All candidate solutions contain 256 final buckets: the flat design clusters the corpus directly into 256 buckets, while an $8\!\times\!32$ hierarchy first forms 8 groups and then 32 buckets per group; $16\!\times\!16$ and $32\!\times\!8$ are defined analogously. For each calibration query, we take the unpruned system's final float-reranked top 10 as the reference and measure what fraction survives pruning.
At matched scan budgets, the flat design retains a larger fraction on every corpus than any hierarchy. Within the flat design, we explore different bucket numbers; increasing the number of retrieved buckets (probes) from 40 to 43 (scan from 18.75\% to 20.16\% of the corpus) increases the four-corpus average containment from 95.58\% to 96.28\%, with further increases leading to diminishing returns. We end up with 256 flat clusters, a capacity factor 1.2, and 43 probes. This scans 16.3--17.7\% real documents before padding and retains 96.2--97.3\% of float NDCG on the three million-scale corpora. At 1\,Gbps, pruning lowers communication by 3.4--4.1$\times$, latency by 3.6--3.9$\times$, and raises throughput by 4.9$\times$ with cloud OT.

\begin{table*}[thbp]
\centering
\caption{Optional-optimization factorial. $R_{10}$ is containment of the full-FR float-reranked top 10; NDCG is relative to full-corpus float search. Slash-separated orders are F+O/F+D/P+O/P+D for communication and throughput, and 100\,Mbps/1\,Gbps/10\,Gbps for latency. Communication is MB/query summed across links; latency is seconds/query; throughput is queries/hour at 1\,Gbps.}
\label{tab:optional}
\vspace{-0.1in}
\small
\setlength{\tabcolsep}{4pt}
\begin{tabular*}{\textwidth}{@{\extracolsep{\fill}} l c c c c c}
\toprule
Dataset & {P quality $R_{10}$/NDCG} & {Communication} & {Full latency} & {Pruned latency} & {Throughput} \\
\midrule
NQ       & 96.8/96.2\% & 244/281/64/71   & 19.91/3.15/1.47 & 4.38/0.80/0.44 & 374/1951/1827/11851 \\
DBpedia  & 92.5/96.2\% & 418/482/104/117 & 34.24/5.42/2.54 & 7.48/1.47/0.87 & 216/1129/1049/6807 \\
Climate  & 97.4/97.3\% & 486/560/117/133 & 39.89/6.32/2.97 & 8.64/1.78/1.09 & 185/966/898/5826 \\
Touch\'e & 99.0/93.9\% & 35.8/41.0/10.4/11.5 & 2.88/0.45/0.21 & 0.68/0.12/0.06 & 2618/13678/12798/83035 \\
\bottomrule
\end{tabular*}
\end{table*}

\begin{figure}[t]
\centering
\includegraphics[width=0.85\linewidth]{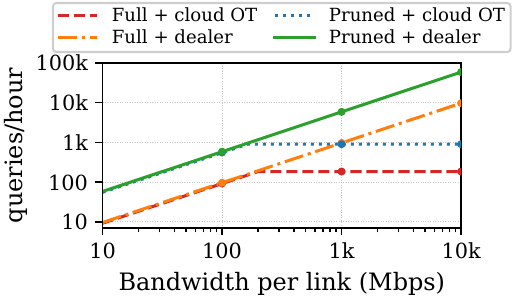}
\vspace{-0.1in}
\Description{Log-log plot of Climate-FEVER queries per hour versus 10 Mbps to 10 Gbps bandwidth for full and pruned scans, each using cloud OT or a trusted dealer. Pruning shifts both curves upward; the dealer removes the high-bandwidth cloud-OT plateau.}
\caption{Climate-FEVER sustained throughput. Markers show 100\,Mbps, 1\,Gbps, and 10\,Gbps. Pruning reduces demand; the dealer raises supply.}
\label{fig:throughput}
\vspace{-0.1in}
\end{figure}

\paragraph{The dealer is lightweight but highly beneficial.} Figure~\ref{fig:throughput} shows sustained throughput across bandwidths for the four optimization combinations. The dealer only changes how the clouds obtain triples and doesn't affect retrieval quality. For a full Climate-FEVER scan, one query requires the dealer to generate 467\,MB of pseudorandom share data and upload a 93.4\,MB correction, which takes 66.1\,ms on one measured CPU core. At 1\,Gbps, replacing cloud Silent OT with the dealer raises sustained throughput from 185 to 966 queries/hour. Supporting that rate requires 125\,MB/s of local PRG expansion and a 201-Mbps uplink, both well below the measured 7.1-GB/s single-core rate and the assumed 1-Gbps link. The organization stores no retrieval index and needs no GPU; if the dealer is unavailable, the clouds can fall back to Silent OT without changing the online circuit.

\paragraph{The optimizations are compatible.} When the clouds generate triples themselves, 48 Silent-OT worker pairs supply 38.39\,M triples/s. A full and pruned Climate-FEVER query consumes 747\,M and 154\,M triples, respectively, so triple generation alone caps their throughput at 185 and 898 queries/hour. At 100\,Mbps, cross-cloud online communication is slower than triple generation; thus, pruning raises throughput by about $6.0\times$ by reducing that traffic, while replacing OT with the dealer adds only another 5\%. 
At 1\,Gbps, however, the link can carry more queries than cloud OT can prepare; pruning alone reaches 898 queries/hour, the dealer alone 966, and with both we can reach 5,826 ($31.5\times$ the unoptimized 185).
At 10\,Gbps, cloud-OT configurations remain capped at the same preprocessing rates, whereas pruning plus the dealer reaches about 58,200 queries/hour before cross-cloud communication becomes the limiting factor. Thus, pruning reduces triple demand and online traffic, while the dealer raises triple supply; either works alone, and their gains compose.


\section{Related Work}
\label{sec:related}

\paragraph{Private Search and RAG} The ownership and trust boundary separates several private-retrieval problems with different objectives. Closest to ours, $p^2$RAG uses two semi-honest non-colluding servers and avoids secure sorting through bisection~\cite{zyskind2024prag,ming2026p2rag}. However, they still score full embeddings and performs iterative secure comparisons; \sys instead makes the representation binary and fixes the radius before the online protocol.

Single-provider systems protect an external querier from the corpus holder under a different corpus-ownership assumption. SANNS and its follow-up PANTHER combine clustering, PIR, secret sharing, garbled circuits, or HE for cryptographic nearest-neighbor search~\cite{chen2020sanns,li2025panther}; Pisces combines oblivious SimHash and BM25 filtering with MPC scoring and PIR-to-share~\cite{liang2026pisces}; and RemoteRAG uses query perturbation plus partially homomorphic scoring over a narrowed search space~\cite{chang2025remoterag}. CipheRAG combines searchable inner-product functional encryption with asymmetric LSH and decryption-enabled attention~\cite{zhou2026cipherag}.
Hua et al.~\cite{hua2026pointing} release a directionally metric-DP learned hash code to form a shortlist, then protect exact candidate reranking with BFV and the final selection with active-secure OT.
These systems keep the corpus at one provider and protect external queries or authorized content access. \sys instead protects the corpus from each of the two outsourced clouds by secret-sharing both stored embeddings and codes.

When the corpus is public or available to the search service, Tiptoe~\cite{henzinger2023tiptoe}, Wally~\cite{asi2024wally}, Speakeasy~\cite{lakshman2026speakeasy}, PACMANN~\cite{zhou2025pacmann}, and PIR-RAG~\cite{wang2025pir} focus on query privacy protection.
In client-owned outsourcing, the owner also queries the corpus. Compass hides HNSW traversal from a malicious server with Ring ORAM, but its adaptive search still requires 8--9 ORAM round trips in the evaluated configurations and $3.2$--$6.8\times$ the plaintext server memory~\cite{zhu2025compass}. MESS avoids ORAM by searching randomized hash codes~\cite{cui2026mess}. The server observes the differentially private perturbed codes, shard assignments, graph topology, traversal traces, and candidate sets. To recover recall, its default configuration routes each item to 16 of 64 HNSW shards, creating $16\times$ indexed-record replication and composing privacy loss across the 16 independently perturbed releases.


\paragraph{Deep hashing} Learned binary codes are a long-standing tool for search efficiency~\cite{luo2023survey,he2025survey,wang2014hashing}; we repurpose them as the MPC-friendly representation and co-design the lightweight training recipe for compact codes. Metric-DP hashing instead deliberately releases a randomized coarse code to support plaintext search at one server. Earlier work applies randomized response to binary codes or establishes extended DP for LSH~\cite{wang2020searching,fernandes2021dplsh}. MESS applies bitwise randomized response to fixed hash codes, obtains extended DP under the induced Hamming pseudometric, and recovers utility with a multi-graph HNSW index~\cite{cui2026mess}.
Hua et al.~\cite{hua2026pointing} randomize the learned continuous pre-sign direction under metric DP, binarize it by post-processing, and use the released code only for shortlisting before encrypted reranking and OT.
These mechanisms address leakage from code intentionally visible to the search provider; in \sys, corpus and query codes are XOR-shared and never released to either server, so no DP is needed.

\section{Conclusion}
Private outsourced retrieval is bottlenecked by the corpus-linear search, and we argue that the bottleneck is set by the representation.
We propose \sys, which combines short learned codes that make the dominant per-document operation a communication-free XOR, a calibrated Hamming radius that removes data-dependent search rounds, and two-server PIR for the final fetch.
On corpora of up to 5.42M documents, the 128-bit configuration retains $95.2$--$97.8\%$ of float NDCG with median candidate sets of 382--1,952; its online filter takes 183\,ms--2.60\,s and reduces latency by three orders of magnitude over int8 cosine MPC.
We further propose two compatible deployment optimizations to achieve higher throughput: private clustering trades minor quality loss for lower demand, while a one-core owner-side dealer accelerates triple supply and falls back to cloud OT. At 1\,Gbps they raise Climate-FEVER throughput individually by $4.9\times$/$5.2\times$ and jointly by $31.5\times$, while the organization stores no retrieval index locally.

\section*{Acknowledgement}
This work was in part supported by the Open Research Fund of State Key Laboratory of Internet Architecture, Tsinghua University, Beijing 100084, China.

\bibliographystyle{ACM-Reference-Format}
\bibliography{references}
\fi

\ifincludeappendix
\ifdefined\fullpaper\fi
\appendix

\section{More Details of Candidate-Generation Protocols}
\label{app:strawmen}
This appendix provides the full pseudocode and accounting for the two strawman protocols of \S\paperref{sec:stage2}, alongside our $\Pi_{\mathrm{FR}}$, as two-party blocks (Figure~\ref{fig:protocols}).

\paragraph{$\Pi_{\mathrm{DFP}}$ (direct full precision).} The data owner additive-shares every stored signed-int8 embedding over $\Z_{2^{32}}$ during offline setup, and the client similarly shares its query. The online protocol therefore starts in the arithmetic domain, batches the $N\!\cdot\!D$ Beaver multiplications into two rounds, and then runs a shared argmax to reveal the top-$k$. The cost is structural and independent of $L$, which is why no hash-side optimization touches it.

\paragraph{$\Pi_{\mathrm{BS}}$ (binary-search radius).} Sharing the popcount of $\Pi_{\mathrm{FR}}$, $\Pi_{\mathrm{BS}}$ then runs $\lceil\log_2(L{+}1)\rceil$ rounds of comparator-plus-count over the shared distances, each revealing the cumulative count below the current radius ($\approx\!2k_{\mathrm{bits}}N$ ANDs for the comparison plus $\approx\!3N$ for the count) and adjusting the radius toward $K$. It is gate-cheaper than materializing a full shared histogram, but every revealed count leaks a sample of the corpus distance CDF, and the $\lceil\log_2(L{+}1)\rceil$ reveal rounds dominate latency on bandwidth-limited links.

\begin{figure*}[t]
\centering
\small
\begin{minipage}[t]{0.47\textwidth}
\centering
\procedure[linenumbering]{$\Pi_{\mathrm{DFP}}$: direct full precision (no hash)}{
\textbf{Hold arithmetic shares } \share{X^{\mathrm a}}_b\!\in\!\Z_{2^{32}}^{N\times D},\ \share{q^{\mathrm a}}_b \\[1pt]
\share{s}_b \gets \share{X^{\mathrm a}}_b\!\cdot\!\share{q^{\mathrm a}}_b \pccomment{$ND$ arith.\ mults; $2$ rounds} \\
\mathcal{K} \gets \mathsf{TopK}(\share{s}_b) \pccomment{shared argmax} \\
\pcreturn\ \mathcal{K}\ (\text{revealed top-}k)
}

\medskip
\procedure[linenumbering]{$\Pi_{\mathrm{FR}}$: calibrated fixed-radius (ours)}{
\textbf{Public } \hat t \text{ from client calibration (Alg.~\paperref{alg:calib})} \\
\textbf{Hold } \share{H}_b\!\in\!\bits^{N\times L},\ \share{q}_b \\[1pt]
\share{E}_b \gets \share{H}_b \oplus (\mathbf{1}_N\!\otimes\!\share{q}_b) \pccomment{free XOR, local} \\
\share{d}_b \gets \popcount(\share{E}_b) \pccomment{$\approx\!LN$ ANDs} \\
\share{m}_b \gets \lepub(\share{d}_b,\hat t) \pccomment{$\approx\!16N$ ANDs} \\
m \gets \mathsf{Reveal}(\share{m}_b) \pccomment{$1$ reveal round} \\
\pcreturn\ \mathcal{K}=\{i:m_i=1\}
}
\end{minipage}\hfill
\begin{minipage}[t]{0.47\textwidth}
\centering
\procedure[linenumbering]{$\Pi_{\mathrm{BS}}$: binary-search radius (data-driven)}{
\textbf{Hold } \share{H}_b\!\in\!\bits^{N\times L},\ \share{q}_b;\ \text{target } K{=}\lceil\rho N\rceil \\[1pt]
\share{E}_b \gets \share{H}_b \oplus (\mathbf{1}_N\!\otimes\!\share{q}_b) \pccomment{free XOR, local} \\
\share{d}_b \gets \popcount(\share{E}_b) \pccomment{$\approx\!LN$ ANDs} \\
(\mathrm{lo},\mathrm{hi}) \gets (0,L) \\
\pcfor r = 1,\dots,\lceil\log_2(L{+}1)\rceil \pcdo \\
\mathrm{mid}\gets\lfloor(\mathrm{lo}{+}\mathrm{hi})/2\rfloor \\
\share{b}_b \gets \mathsf{LE}(\share{d}_b,\mathrm{mid}) \pccomment{$\approx\!2k_{\mathrm{bits}}N$ ANDs} \\
c \gets \mathsf{Reveal}(\textstyle\sum_i \share{b_i}_b) \pccomment{$\approx\!3N$ ANDs; \textbf{leaks} CDF} \\
\text{adjust } (\mathrm{lo},\mathrm{hi}) \text{ by } [\,c \ge K\,] \\
\pcendfor \pccomment{$\lceil\log_2(L{+}1)\rceil$ reveal rounds} \\
\pcreturn\ \mathcal{K}=\{i:d_i\le \mathrm{hi}\}
}
\end{minipage}
\caption{The three candidate-generation protocols as two-party blocks (party $P_b$'s view; inline comments are the dominant per-query \emph{online} cost). $\Pi_{\mathrm{DFP}}$ is $L$-independent arithmetic work over embeddings that are additive-shared during offline setup. $\Pi_{\mathrm{BS}}$ and $\Pi_{\mathrm{FR}}$ share the popcount (free XOR followed by a Wallace tree) and differ only in candidate selection: \textsc{BS} binary-searches the radius, spending $\lceil\log_2(L{+}1)\rceil$ data-dependent count reveals that each leak a sample of the corpus distance CDF, whereas \textsc{FR} consumes a client-calibrated public radius $\hat t$ and collapses candidate selection to one comparison and one reveal of the indicator $m$.}
\label{fig:protocols}
\end{figure*}

\paragraph{DFP step breakdown.} Table~\ref{tab:dfp} decomposes $\Pi_{\mathrm{DFP}}$ on the arithmetic engine at $D{=}768$. Batching all independent products makes the inner product two online rounds; the remaining rounds come from the repeated shared argmax. Above $N{=}256$, total wall-clock remains within $6\%$ of $562\,\mu$s/doc and communication converges to $87.9$\,KB/doc. We therefore extrapolate DFP's columns in Tables~\paperref{tab:motivation} and~\paperref{tab:protocols} from the measured $N{=}4096$ anchor: materializing the arithmetic triples for a BEIR-scale run is infeasible on this machine. The cost remains structural—$N\!\cdot\!D$ arithmetic multiplications followed by shared top-$k$, with neither term involving $L$.

\begin{table}[htbp]
\centering
\caption{DFP baseline step breakdown (arithmetic MPC, $D{=}768$, $k{=}10$, online-only). Corpus embeddings and queries are additive-shared during offline setup; the table includes the batched inner products and shared top-$k$.}
\label{tab:dfp}
\small
\setlength{\tabcolsep}{4.5pt}
\begin{tabular}{S[table-format=4.0] S[table-format=1.3] S[table-format=1.3] S[table-format=1.3] S[table-format=3.2]}
\toprule
{$N$} & {cosine (s)} & {top-$k$ (s)} & {total (s)} & {MB} \\
\midrule
64   & 0.042 & 0.014 & 0.058 & 5.59 \\
256  & 0.083 & 0.050 & 0.136 & 22.46 \\
1024 & 0.323 & 0.204 & 0.543 & 89.94 \\
4096 & 1.390 & 0.853 & 2.302 & 359.89 \\
\bottomrule
\end{tabular}
\end{table}

\section{Leakage: Formal Simulation Security and Relational Analysis}
\label{app:leakage}

This appendix defines the setup and query leakage of each deployed variant, proves adaptive multi-query simulation security, and quantifies the relational information contained in the full-scan access pattern.

\subsection{Experiments and Leakage Functions}
Let $\lambda$ be the security parameter and let the logical database be $\mathsf{DB}=(H,E,W)$, where $H\in\bits^{N\times L}$ contains binary codes, $E\in\bits^{N\times B_E}$ contains serialized fixed-width embedding rows, and $W=(W_i)_{i\in[N]}$ contains content rows. The setup algorithm samples $\pi\leftarrow S_N$ independently of $\mathsf{DB}$, writes $\widetilde H_i=H_{\pi(i)}$, $\widetilde E_i=E_{\pi(i)}$, pads every $W_{\pi(i)}$ to the derived payload width $B_{\mathrm{pt}}$, and encrypts it with a distinct nonce into a $B_{\mathrm{ct}}$-byte row. For $b\in\{A,B\}$, server $P_b$ receives
\begin{align*}
 \mathsf{st}_b&=\big(\share{\widetilde H}_b,\share{\widetilde E}_b,C,\mathsf{pp}_v\big),\\
 C_i&\gets\mathsf{Enc}_{K_{\mathrm{enc}}}
   \big(\mathsf{nonce}_i,\mathsf{pad}_{B_{\mathrm{pt}}}(W_{\pi(i)})\big),
 \qquad |C_i|=B_{\mathrm{ct}}.
\end{align*}
where $v\in\{\mathrm{full},\mathrm{prune}\}$ selects the protocol variant and $\mathsf{pp}_v$ contains its public dimensions and circuit parameters. The client retains $K_{\mathrm{enc}}$ and $\pi$.

\begin{definition}[Setup leakage]
Define the public parameter tuples
\begin{align*}
 \mathsf{pp}_{\mathrm{full}}&=(N,L,B_E,B_{\mathrm{ct}},\hat t,k),\\
 \mathsf{pp}_{\mathrm{prune}}&=(\mathsf{pp}_{\mathrm{full}},C_{\mathrm{clust}},p,B_{\mathrm{bucket}}).
\end{align*}
The setup leakage is
\[
 \mathcal{L}^{v}_{\mathrm{stp}}(\mathsf{DB})=\mathsf{pp}_v.
\]
The permutation $\pi$, plaintext lengths before padding, hash-prefix order, and cluster membership are absent from setup leakage.
\end{definition}

For a full scan and query code $q$, define
\[
 d_{q,i}=\operatorname{HW}(q\oplus\widetilde H_i),\qquad
 \mathcal K_q=\{i\in[N]:d_{q,i}\le\hat t\}.
\]

\begin{definition}[Full-scan query leakage]
\label{def:leakage-full}
The base fixed-radius query leakage is
\[
 \mathcal L^{\mathrm{full}}_{\mathrm{qry}}(\mathsf{DB},q)=(\hat t,\mathcal K_q).
\]
It reveals stable permuted slot identities, hence $|\mathcal K_q|$ and equality of response sets. Distinct queries may induce the same response set.
\end{definition}

Private pruning retrieves $p$ padded buckets into freshly randomized XOR shares. Let $M=pB_{\mathrm{bucket}}$ be the public buffer length, let $R_q\in(\bits^L\cup\{\bot\})^M$ be the hidden ordered buffer, and define $\mu_q[\ell]=[R_q[\ell]\ne\bot\wedge\operatorname{HW}(q\oplus R_q[\ell])\le\hat t]$.

\begin{definition}[Pruned-scan query leakage]
\label{def:leakage-prune}
The private-pruning leakage is
\[
 \mathcal L^{\mathrm{prune}}_{\mathrm{qry}}(\mathsf{DB},q)=(\hat t,M,\mu_q).
\]
The selected bucket identifiers and the map from buffer positions to persistent corpus slots remain hidden. Fresh PIR selector shares and fresh mask-correction shares prevent a server from linking a buffer position across queries.
\end{definition}

The complete leakage for an adaptively generated sequence $q_1,\ldots,q_Q$ is
\[
 \mathcal L^v(\mathsf{DB};q_1,\ldots,q_Q)=
 \left(\mathcal L^v_{\mathrm{stp}}(\mathsf{DB}),
       \big(\mathcal L^v_{\mathrm{qry}}(\mathsf{DB},q_j)\big)_{j=1}^{Q}\right).
\]

For comparison, binary-search selection uses $T=\lceil\log_2(L+1)\rceil$ public thresholds $\tau_{q,1},\ldots,\tau_{q,T}$ and opens $c_{q,r}=\sum_i[d_{q,i}\le\tau_{q,r}]$ before revealing its final set $\mathcal K_q^{\mathrm{BS}}$. Its query leakage is
\[
 \mathcal L^{\mathrm{BS}}_{\mathrm{qry}}(q)=
 \left((\tau_{q,r},c_{q,r})_{r=1}^{T},\tau_q^\star,\mathcal K_q^{\mathrm{BS}}\right).
\]
When FR and BS are conditioned to produce the same final pair $(\tau_q^\star,\mathcal K_q)$, FR leakage is the projection of BS leakage that deletes the intermediate count pairs. Without this conditioning, the two protocols implement different selection functions and their leakage tuples are not ordered by set inclusion.

\subsection{Simulation Security}
For $P\in\{P_A,P_B\}$, let $\mathsf{REAL}^{v}_{P}(1^\lambda,\mathsf{DB},\mathbf q)$ be $P$'s complete state, random tape, preprocessing transcript, received messages, sent messages, and opened values in a real execution on the adaptively generated query sequence $\mathbf q$. Let $\mathsf{IDEAL}^{v}_{P,\mathcal S}(1^\lambda,\mathcal L^v)$ be the output of a simulator receiving leakage online in the same order.

We use four standard assumptions. (A1)~$\mathsf{Enc}$ is multi-message IND-CPA secure for distinct nonces, and the PRG masking the private-pruning bucket database is pseudorandom. (A2)~The Boolean protocol securely realizes its deterministic circuit against one semi-honest corrupted server. (A3)~Silent-OT preprocessing or the seeded dealer securely realizes independent one-time Beaver triples; the dealer PRG is secure and no triple is reused. (A4)~The two-server XOR-PIR uses independent uniform selector shares and non-colluding servers.

\begin{theorem}[Adaptive multi-query simulation]
\label{thm:sim}
Under (A1)--(A4), for every $v\in\{\mathrm{full},\mathrm{prune}\}$, every $P\in\{P_A,P_B\}$, and every polynomially bounded adaptive query sequence $\mathbf q$, there exists a PPT simulator $\mathcal S^v_P$ such that
\begin{multline*}
 \left\{\mathsf{REAL}^{v}_{P}(1^\lambda,\mathsf{DB},\mathbf q)\right\}_{\lambda\in\mathbb N}
 \stackrel{c}{\approx}\\
 \left\{\mathsf{IDEAL}^{v}_{P,\mathcal S^v_P}
 \big(1^\lambda,\mathcal L^v(\mathsf{DB};\mathbf q)\big)\right\}_{\lambda\in\mathbb N}.
\end{multline*}
\end{theorem}

\begin{proof}
Fix $P=P_A$; symmetry gives the construction for $P_B$. Define hybrids $\mathsf H_0,\ldots,\mathsf H_5$ over the complete multi-query view.

\emph{$\mathsf H_0$.} This is $\mathsf{REAL}^{v}_{P_A}(1^\lambda,\mathsf{DB},\mathbf q)$.

\emph{$\mathsf H_1$.} Replace every $C_i=\mathsf{Enc}_{K_{\mathrm{enc}}}(\mathsf{nonce}_i,\mathsf{pad}_{B_{\mathrm{pt}}}(W_{\pi(i)}))$ by $C_i^0=\mathsf{Enc}_{K_{\mathrm{enc}}}(\mathsf{nonce}_i,0^{B_{\mathrm{pt}}})$. For $v=\mathrm{prune}$, also replace the PRG-masked padded bucket database by an equal-length uniform string. A standard multi-message encryption hybrid followed by a PRG hybrid and (A1) give $\mathsf H_0\stackrel{c}{\approx}\mathsf H_1$.

\emph{$\mathsf H_2$.} Generate $\share{\widetilde H}_{A}\leftarrow\bits^{N\times L}$ and $\share{\widetilde E}_{A}\leftarrow\bits^{N\times B_E}$ uniformly, and for each query generate $\share{q}_A\leftarrow\bits^L$ uniformly. This changes no distribution: in XOR sharing, either share is uniform for every fixed secret. Sample one persistent pair $(\share{\widetilde H}_{A},\share{\widetilde E}_{A})$ and reuse it throughout the sequence, thereby preserving equality and overlap among all messages derived from the same stored row.

\emph{$\mathsf H_3$.} Replace the Silent-OT or dealer preprocessing view by the simulator of (A3), maintaining a monotone counter so every nonlinear gate consumes a distinct simulated triple. Sequential composition over the polynomial number of gates and queries yields $\mathsf H_2\stackrel{c}{\approx}\mathsf H_3$.

\emph{$\mathsf H_4$.} Process queries in arrival order. For query $q_j$, invoke the simulator guaranteed by (A2) on the corrupted party's sampled input shares and the clear output prescribed by $\mathcal L^v_{\mathrm{qry}}$. For $v=\mathrm{full}$, this output is the indicator of $\mathcal K_{q_j}$; for $v=\mathrm{prune}$, it is $\mu_{q_j}$. Adaptive sequential composition applies because the next query may depend on earlier leakage but each circuit invocation uses fresh triples and a fixed public schedule. Hence $\mathsf H_3\stackrel{c}{\approx}\mathsf H_4$.

\emph{$\mathsf H_5$.} For every bucket or content PIR invocation, sample $P_A$'s selector share uniformly. Compute its response by applying the specified XOR-linear PIR algorithm to that selector and the simulated persistent database $C^0$ or simulated masked bucket database. Thus responses retain their exact algebraic dependence on the database; they are not replaced by independent uniform strings. For private pruning, sample the client's fresh mask-correction share uniformly and derive the resulting transient buffer share. By (A4), the selector distribution is identical to the real one and is independent of the selected index. In a full scan, candidate-embedding messages are read from the single persistent simulated share array at the positions prescribed by leakage, preserving repeated-row correlations. Under pruning, they are read from the freshly randomized transient buffer at positions prescribed by $\mu_q$. Therefore $\mathsf H_4\equiv\mathsf H_5$ for PIR privacy and simulated storage, up to the component replacements already made.

The simulator $\mathcal S^v_{P_A}$ implements $\mathsf H_5$ using only $\mathcal L^v$: it samples persistent shares, zero ciphertexts, and preprocessing state at setup, then extends the same state for every online leakage tuple. Consequently $\mathsf H_0\stackrel{c}{\approx}\mathsf H_5=\mathsf{IDEAL}^{v}_{P_A,\mathcal S^v_{P_A}}$.
\end{proof}

The theorem applies to one corrupted server. Colluding servers reconstruct the XOR-shared codes and embeddings. The replicated content remains confidential against their collusion under (A1) because neither server holds the encryption key.

\subsection{Ciphertext-only Relational Inference}
Each full-scan set $\mathcal K_q$ is a Hamming ball around an unknown query. Repeated balls induce a kernel-blurred proximity statistic over stable encrypted slot identifiers. Tight radii produce sharper conditional neighborhoods over fewer observed slots; loose radii produce broader, less resolved neighborhoods. Passive workload support bounds the observation: a slot absent from every candidate set has no incidence edge.

\paragraph{Additional empirical assumption.} The following attack experiment adds an auxiliary-information restriction beyond the cryptographic theorem: the observing server has no plaintext or encoded reference corpus, no slot-to-document mapping, and no known query-to-plaintext anchors. This assumption isolates ciphertext-only relational leakage. Known-data, similar-data, and known-query attacks require separate auxiliary inputs and are discussed below.

\paragraph{Estimator and evaluation universe.} For workload $\mathcal Q$, define the incidence matrix $X\in\bits^{|\mathcal Q|\times N}$ by $X_{q,i}=[i\in\mathcal K_q]$, the frequency $f_i=\sum_qX_{q,i}$, the observed universe $P$, and the anchor-eligible universe $U$:
\[
 P=\{i\in[N]:f_i\ge1\},\qquad U=\{i\in P:f_i\ge2\}.
\]
For distinct $i,j\in P$, the estimator is the cosine of their incidence columns,
\[
 M_{ij}=\sum_qX_{q,i}X_{q,j},\qquad
 s(i,j)=\frac{M_{ij}}{\sqrt{f_if_j}}.
\]
For each of five seeds (42--46), we independently sample the data-independent slot permutation and a uniform anchor set $S\subseteq U$ of size $\min(1500,|U|)$; Webis-Touch\'e at $L=256$ uses all $640$ eligible slots. For each $i\in S$, $\widehat\Gamma_{10}(i)$ contains the ten highest-scoring elements of $P\setminus\{i\}$, and $\Gamma^{H}_{10}(i)$ contains the ten nearest elements of the same observed universe under true Hamming distance. Both rankings break ties by ascending permuted slot identifier. We report
\[
 \operatorname{P@10}=\frac1{|S|}\sum_{i\in S}
 \frac{|\widehat\Gamma_{10}(i)\cap\Gamma^H_{10}(i)|}{10}.
\]
The condition-specific chance baseline is $10/(|P|-1)$ for $|P|>10$. The popularity null replaces $\widehat\Gamma_{10}(i)$ by the ten highest-$f_j$ elements of $P\setminus\{i\}$ and is scored by the same formula. Corpus observation coverage is $|P|/N=|\cup_q\mathcal K_q|/N$; P@10 is conditional on anchors in $S\subseteq U$.

We use two workloads. \emph{Held-out} uses the BEIR test queries excluded from radius calibration: 3,352 for NQ, 300 for DBpedia, 1,435 for Climate-FEVER, and 39 for Webis-Touch\'e. The deployed $L=128/256$ radii are respectively $35/72$, $34/70$, $34/72$, and $28/50$; the corresponding held-out $|U|$ values are $1{,}516{,}705/1{,}030{,}056$, $86{,}900/28{,}602$, $50{,}504/34{,}044$, and $2{,}463/640$. \emph{Coverage stress} samples 20,000 corpus rows without replacement as query centers and evaluates the resulting passive transcript; the stress workload increases coverage but does not grant query-injection capability to the server.

\paragraph{Results.} Table~\ref{tab:leak-attack} shows three effects. At $L=128$, held-out P@10 is $0.017$--$0.070$, establishing a measurable unlabeled edge signal. Coverage controls its corpus scope: NQ's held-out workload observes $76.0\%$ of slots, whereas Climate-FEVER observes $2.4\%$. Under coverage stress, P@10 reaches $0.099$--$0.280$ at $L=128$ and $0.082$--$0.245$ at $L=256$. Longer codes reduce every stress result, while tighter balls can increase conditional held-out precision over a smaller observed region. Across all cells, the five-seed standard deviation is at most $0.0073$.

\begin{table*}[t]
\centering
\caption{Ciphertext-only Hamming-neighborhood inference at the deployed radius. Each cell reports five-seed P@10 mean$\pm$std / corpus coverage $|P|/N$. Anchors lie in $U=\{i:f_i\ge2\}$ and neighbors in $P=\{i:f_i\ge1\}$; chance and popularity use the same $P$.}
\label{tab:leak-attack}
\small
\setlength{\tabcolsep}{3.5pt}
\begin{tabular}{@{}l r c c c c@{}}
\toprule
 & & \multicolumn{2}{c}{$L{=}128$} & \multicolumn{2}{c}{$L{=}256$} \\
\cmidrule(lr){3-4}\cmidrule(lr){5-6}
Corpus & $N$ & held-out & stress & held-out & stress \\
\midrule
NQ              & 2.68M & $.068{\pm}.002$/76\% & $.200{\pm}.003$/98\% & $.072{\pm}.002$/61\% & $.170{\pm}.004$/94\% \\
DBpedia         & 4.64M & $.020{\pm}.002$/11\% & $.104{\pm}.002$/95\% & $.029{\pm}.001$/4.9\% & $.082{\pm}.005$/83\% \\
Climate-FEVER   & 5.42M & $.070{\pm}.002$/2.4\% & $.099{\pm}.005$/95\% & $.101{\pm}.003$/1.4\% & $.088{\pm}.002$/89\% \\
Webis-Touch\'e  & 383K  & $.017{\pm}.002$/9.4\% & $.280{\pm}.007$/92\% & $.041{\pm}.007$/6.8\% & $.245{\pm}.005$/75\% \\
\bottomrule
\end{tabular}
\end{table*}

\subsection{Hidden Candidate Padding}
\label{app:leakage_defense}
Padding protects membership only if neither server first receives the unpadded indicator. The protected variant changes the Filter--Rerank boundary as follows. After computing shared indicators $\share{m}_A,\share{m}_B$, each server sends its share directly to the client; the servers do not reconstruct $m$. The client reconstructs $\mathcal K_q$, samples
\begin{align*}
 s_q&=\min\{\lceil r|\mathcal K_q|\rceil,N-|\mathcal K_q|\},\\
 D_q&\xleftarrow{\$}
 \{D\subseteq[N]\setminus\mathcal K_q:|D|=s_q\}.
\end{align*}
The client sends only $\mathcal K_q^+=\mathcal K_q\cup D_q$ to both servers. They return embedding shares for $\mathcal K_q^+$; the client removes $D_q$ before reranking. A single server's query leakage becomes
\[
 \mathcal L^{\mathrm{pad}(r)}_{\mathrm{qry}}(q)=(\hat t,\mathcal K_q^+).
\]
Ratio padding preserves noisy membership but reveals $|\mathcal K_q|$ up to deterministic rounding through $|\mathcal K_q^+|$ and public $r$. A fixed-target variant samples $B_{\mathrm{pad}}-|\mathcal K_q|$ decoys for a public $B_{\mathrm{pad}}\ge|\mathcal K_q|$ and thereby fixes response length; queries exceeding $B_{\mathrm{pad}}$ use an explicitly declared overflow policy.

\begin{corollary}[Hidden-padding simulation]
\label{cor:pad-sim}
Let
\[
 \mathcal L^{\mathrm{pad}(r)}(\mathsf{DB};\mathbf q)=
 \left(\mathcal L^{\mathrm{full}}_{\mathrm{stp}}(\mathsf{DB}),
       \big(\mathcal L^{\mathrm{pad}(r)}_{\mathrm{qry}}(q_j)\big)_{j=1}^{Q}\right).
\]
Under (A1)--(A4), the conclusion of Theorem~\ref{thm:sim} holds for the hidden-padding protocol with leakage $\mathcal L^{\mathrm{pad}(r)}$.
\end{corollary}

\begin{proof}
In hybrid $\mathsf H_4$, simulate the corrupted server's circuit view with no clear server output; its share sent to the client is uniform. Supply $\mathcal K_q^+$ from query leakage as the subsequent client message, and read the corresponding embedding shares from the persistent simulated array. The remaining hybrids are unchanged.
\end{proof}

Table~\ref{tab:leak-pad} evaluates exactly this server-visible union under the $L=128$ coverage stress. At $r=2$, P@10 falls by $39$--$57\%$ across the four corpora, while the mean revealed fraction is $0.24$--$1.10\%$. This experiment quantifies resistance to the stated estimator; access-pattern obfuscation and search-pattern leakage remain separate dimensions~\cite{islam2012access,oya2021hiding}. Stronger accumulation control requires an owner-driven epoch change that samples a fresh secret slot permutation, re-shares the index, and re-encrypts rows, or a leakage-suppression construction that hides response linkage across epochs~\cite{kamara2018structured}.

\begin{table*}[t]
\centering
\caption{Hidden-padding dose--response at $L{=}128$ under the 20k-query coverage stress. Each cell reports mean revealed fraction $|\mathcal K_q^+|/N$ / five-seed P@10 mean$\pm$std using the corresponding observed and anchor-eligible universes.}
\label{tab:leak-pad}
\small
\setlength{\tabcolsep}{5pt}
\begin{tabular}{@{}l c c c c@{}}
\toprule
pad ratio $r$ & $0$ & $0.5$ & $1$ & $2$ \\
\midrule
NQ              & $.09\%/.200{\pm}.003$ & $.14\%/.146{\pm}.005$ & $.19\%/.129{\pm}.002$ & $.28\%/.123{\pm}.003$ \\
DBpedia         & $.09\%/.104{\pm}.002$ & $.13\%/.075{\pm}.002$ & $.17\%/.067{\pm}.004$ & $.26\%/.063{\pm}.004$ \\
Climate-FEVER   & $.08\%/.099{\pm}.005$ & $.12\%/.067{\pm}.003$ & $.16\%/.064{\pm}.004$ & $.24\%/.054{\pm}.002$ \\
Webis-Touch\'e  & $.37\%/.280{\pm}.007$ & $.55\%/.169{\pm}.002$ & $.73\%/.148{\pm}.002$ & $1.10\%/.120{\pm}.002$ \\
\bottomrule
\end{tabular}
\end{table*}

\subsection{Relation to Leakage-abuse Attacks}
IKK and Count use access-pattern co-occurrence together with known or sampled plaintext information to recover query labels~\cite{islam2012access,cash2015leakage}. Refined score attacks use a distributionally similar corpus and a small set of known query anchors~\cite{damie2021highly}; frequency-based attacks additionally exploit the query distribution and search pattern~\cite{oya2021hiding}. Under the additional empirical assumption above, these auxiliary graphs and anchors are unavailable, so our experiment measures unlabeled edge inference rather than semantic query recovery. Partial known-data, similar-data, and known-query regimes can attach semantics to the recovered relation and constitute strictly richer evaluations.

Per-document volume attacks~\cite{blackstone2020revisiting} exploit ciphertext byte lengths. Our fixed-width padding removes that input at the cost of $NB_{\mathrm{ct}}$ content storage. The remaining server-visible objects are the public row count and the access-pattern leakage specified above.

\section{PIR Fetch Overhead}
\label{app:fetch}

The \textsf{Fetch} step—two-server PIR over the replicated content ciphertext (\S\paperref{sec:fetch})—performs \emph{no} secure computation: zero AND gates, zero Beaver triples, and one network round. Its cost is the PIR query and ciphertext response plus a local plaintext XOR-reduce at each server over the $|\mathcal{K}|$ ciphertext rows. Table~\ref{tab:fetch-overhead} applies this model to the same four $L{=}128$ operating points as the main evaluation. For a 10-KB blob, client-facing traffic is 206--210\,KB, only $0.04$--$0.63\%$ of the \textsf{Filter} traffic. The local scan adds $0.15$--$2.16\%$ latency at 10\,KB and $1.51$--$21.5\%$ at 100\,KB; Webis-Touch\'e has the largest ratio because its corpus-linear \textsf{Filter} is the smallest. The estimate uses $2k|\mathcal{K}|B_{\mathrm{ct}}$ bytes at 20\,GB/s, while the \textsf{Filter} columns are measured or scaled exactly as in Table~\paperref{tab:protocols}.

\begin{table*}[thbp]
\centering
\caption{Per-query PIR \textsf{Fetch} overhead at the four deployed $L{=}128$ operating points ($k{=}10$). ``net'' uses a 10-KB blob and reports bytes / percent of \textsf{Filter} traffic. The last columns estimate local scan latency as a percent of \textsf{Filter} latency for three blob sizes.}
\label{tab:fetch-overhead}
\small
\setlength{\tabcolsep}{6pt}
\begin{tabular}{@{}l r r r r l c c c@{}}
\toprule
& & & \multicolumn{2}{c}{\textsf{Filter}} & & \multicolumn{3}{c}{\textsf{Fetch} latency (\% F)} \\
\cmidrule(lr){4-5}\cmidrule(lr){7-9}
Corpus & $N$ & $K_{50}$ & comm (MB) & lat (ms) & net @ 10\,KB & 1\,KB & 10\,KB & 100\,KB \\
\midrule
NQ & 2{,}681{,}468 & 1{,}952 & 230.61 & 1284.8 & 210.2\,KB (0.09\%) & 0.16\% & 1.56\% & 15.6\% \\
DBpedia & 4{,}635{,}922 & 1{,}207 & 398.70 & 2221.3 & 208.4\,KB (0.05\%) & 0.06\% & 0.56\% & 5.57\% \\
Climate-FEVER & 5{,}416{,}593 & 382 & 465.83 & 2595.4 & 206.3\,KB (0.04\%) & 0.02\% & 0.15\% & 1.51\% \\
Webis-Touch\'e & 382{,}545 & 385 & 32.90 & 183.3 & 206.3\,KB (0.63\%) & 0.22\% & 2.16\% & 21.5\% \\
\bottomrule
\end{tabular}
\end{table*}

\section{Deep Hash Training: Configuration and Design Space}
\label{app:hash-training}
\label{sec:ablation}
 
This appendix expands \S\paperref{sec:hash-arch} with the details for reproducing the hash model training.
 
\subsection{Training Configuration}
\label{app:trainconfig}
Table~\ref{tab:trainconfig} lists the full configuration, shared across all code widths; only $L$ and the epoch budget vary across checkpoints. We adapt \texttt{e5-base-v2} with LoRA on all linear layers and train the linear hash head jointly, under constant learning rates and constant loss weights, with no warmup or learning-rate schedule. Hard negatives are mined per query by BM25 (top-$512$)~\cite{robertson1994bm25} reranked by a MiniLM cross-encoder~\cite{wang2020minilm} into a pool, of which $m{=}3$ participate in each step's backward pass. Mining is cached per query id and refreshed every $4$ epochs with lazy first-touch population, which cuts cross-encoder forwards per epoch by an order of magnitude on a long-tailed query distribution.
 
\begin{table}[t]
\centering
\caption{Training configuration, shared across all code widths.}
\label{tab:trainconfig}
\small
\setlength{\tabcolsep}{6pt}
\resizebox{\columnwidth}{!}{
\begin{tabular}{l l}
\toprule
Encoder & \texttt{e5-base-v2} (768d), mean-pool, $\ell_2$-norm \\
Adaptation & LoRA, all-linear, $r{=}16$, $\alpha{=}16$, dropout $0.05$ \\
Hash head & linear $\to L$, $\tanh(\beta\cdot)$, $\beta:1{\to}6$ \\
Optimizer & AdamW, wd $0.01$, no schedule, no warmup \\
Learning rate & encoder $2\!\times\!10^{-6}$, head $2\!\times\!10^{-4}$ \\
Batch & $128$ queries $\times$ ($1$ pos $+\,3$ hard neg), $300$ steps/epoch \\
Epochs & $16$ (default); $24$ for the longer-training ablation \\
Losses & $\mathcal{L}_{\mathrm{nce}}(T{=}0.05) + 0.8\,\mathcal{L}_{\mathrm{bin}}(\tau_b{=}0.1) + 1.0\,\mathcal{L}_{\mathrm{dist}}$ \\
Disabled & $\lambda_{\mathrm{bal}}{=}\lambda_{\mathrm{quant}}{=}\lambda_{\mathrm{ind}}{=}0$ (App.~\ref{app:regularizers}) \\
Hard negs & BM25-512 $\to$ MiniLM rerank, cache refresh / $4$ ep \\
\bottomrule
\end{tabular}
}
\end{table}
 
\subsection{Soft-to-hard Schedule}
\label{app:annealing}
The sign function that produces the final code is not differentiable, so we train on a smooth surrogate and harden it over the run. The head's logits are squashed by $\tanh(\beta\cdot\mathrm{logit})$ with the inverse-temperature $\beta$ annealed linearly from $1$ to $6$ across training: early on, the soft code is smooth and carries useful gradient everywhere, and late in training, it concentrates near $\pm1$ so the soft-to-hard gap closes. At inference the surrogate is dropped and the code is simply $\sgn(\mathrm{logit})$, which is exactly what the offline indexing step signs to obtain $H\in\bits^{N\times L}$ (\S\paperref{sec:protocol}).
 
\subsection{Regularizers}
\label{app:regularizers}
A deep-hashing objective is conventionally more than a ranking loss. Around the relevance term, the literature accumulates a set of \emph{code-quality regularizers} that push the relaxed (e.g., $\tanh$) outputs toward well-behaved binary codes, and most published systems carry two or three of them at once~\cite{luo2023survey,he2025survey,wang2014hashing}. For example, spectral relaxations of the binary-code objective showed that good codes should be \emph{balanced} (each bit splits the corpus evenly) and \emph{uncorrelated} (distinct bits carry independent information), and where iterative-quantization analyzes showed that the rounding step from real vectors to bits should incur as little \emph{quantization error} as possible~\cite{wang2014hashing}. Recent surveys reorganize the deep-era versions of these around the same three axes: few-bit/compact codes, code balance, and low quantization error~\cite{he2025survey,luo2023survey}. On a soft code $b\in[-1,1]^{L}$ over a batch of $m$ examples, the three classic terms read:
 
\begin{itemize}
\item \emph{Quantization} $\mathcal{L}_{q}=\big\| |b|-\mathbf{1}\big\|_1$, which penalizes logits near zero and drives each soft value toward a confident $\pm1$, shrinking the gap between the soft code optimized at training time and the hard code emitted at inference. This is the most widely used of the three: it appears as an $\ell_1$ or $\ell_2$ penalty in the pairwise-likelihood line~\cite{li2015feature}, is recast as a bimodal-Laplacian prior on the code in DHN~\cite{zhu2016deep}, is sharpened into a Cauchy/margin form to concentrate mass inside small Hamming radii in MMHH and Deep Fisher Hashing~\cite{kang2019maximum,li2019push}, and is the explicit reconstruction objective of the quantization-based family~\cite{liu2018deep,yuan2020central}.
\item \emph{Bit-balance} $\mathcal{L}_{\mathrm{bal}}=\big|\mathbf{1}^{\top} b\big|$, which penalizes a bit that takes the same sign across the batch—a constant bit carries no information and wastes one of the $L$ slots. It descends directly from the balance constraint of spectral hashing and is carried into deep models as an explicit term~\cite{do2016bdnn}, or engineered away by construction—e.g.\ a batch-normalization or bi-half layer that forces an even split without a tunable weight~\cite{hoe2021one}.
\item \emph{Bit-independence} $\mathcal{L}_{\mathrm{ind}}=\big\|b^{\top}b/m - I\big\|$, which suppresses off-diagonal correlations so the code does not spend several bits encoding the same direction. It is the deep analog of the uncorrelated-bit constraint from classical hashing and typically travels together with the balance term~\cite{do2016bdnn,wang2014hashing}; a parallel line sidesteps it by mapping classes to mutually orthogonal target codes (Hadamard or hash-center constructions) so independence holds by design rather than by penalty~\cite{yuan2020central,wang2023deep,shen2018nash}.
\end{itemize}
 
The cost of this machinery is well documented: each term adds a loss weight, and the combined objective is delicate—the extra penalties introduce \textit{more} hyperparameters to tune, and the numerical optimization is prone to poor local minima, which is precisely the motivation behind recent ``single-loss'' designs that fold balance and quantization back into one ranking-style objective~\cite{hoe2021one}.

We take the same lesson to its conclusion and set $\lambda_{q}=\lambda_{\mathrm{bal}}=\lambda_{\mathrm{ind}}=0$, relying on the listwise margin and the teacher anchor alone. The justification is empirical: the active objective already places the codes in the regime those penalties target, potentially due to the extensive pretraining and tuning already inherent to the embedding geometry of the encoder itself~\cite{wang2022e5}.

In particular, a diagnostic at mid-training shows the bits are well balanced and the logits are near-saturated under the bare objective—per-bit entropy $\approx\!0.99$ and mean absolute bit activation $\approx\!0.07$, which is the operating point that an explicit balance or quantization term is meant to reach.
The mechanism is that the listwise margin supplies the saturation pressure for free: ranking the positive above the entire negative pool in the inner-product (Hamming) geometry requires confident, well-separated codes, which pushes logits away from the sign boundary as a side effect, leaving the quantization penalty little to do; and InfoNCE on $\ell_2$-normalized embeddings spreads probability mass across directions \cite{wang2020uniformity}, which discourages the constant or duplicated bits that the balance and independence terms exist to remove. Consistent with this, adding any of the three penalties changes the codes negligibly, while in controlled comparisons at matched float quality, the bare objective yields \emph{more} discriminative codes than with any one of them dialed in, broadly mirroring the move in the hashing literature away from multi-term recipes~\cite{hoe2021one,he2025survey}.

\subsection{Design Space}
\label{app:design-space}
Several axes of the hash model admit alternatives; we summarize the choices and the reasoning, and note that the broader sweep informs but does not gate the protocol.
 
\emph{Encoder adaptation.} We compared full fine-tuning, freezing the encoder, unfreezing only the top layers, and LoRA. Full fine-tuning degrades zero-shot transfer by overwriting the pretrained geometry the rerank relies on; freezing leaves too little plasticity to reshape the space for Hamming retrieval. LoRA sits between the two—enough capacity to specialize the code while keeping the float embedding anchored—and gives the best end-to-end hybrid quality, so it is our default.
 
\emph{Hash head.} A single linear projection from the pooled embedding to $L$ logits is sufficient because the LoRA-adapted encoder already does the representational work; a deeper non-linear head adds parameters without a corresponding quality gain in our setting.
 
\emph{Code length.} Quality and candidate concentration improve smoothly from 96 to 128 bits and then flatten toward 256. Because every bit is a recurring AND-gate cost, we operate at the knee, $L{=}128$, rather than past it; we report $L\in\{96,128,256\}$ to show the trade explicitly (\S\paperref{sec:eval-quality}).
 
\emph{Hard-negative mining.} The two-stage BM25-then-cross-encoder miner supplies the gradient signal for the listwise loss; caching and periodic refresh keep its cost off the per-step critical path. Weaker mining (BM25 alone) measurably softens the binary ranking margin, which is why the cross-encoder stage is retained despite its offline cost.

\ifdefined\appendixonly
\bibliographystyle{ACM-Reference-Format}
\bibliography{references}
\fi
\fi

\end{document}